\documentclass[11pt,a4paper]{article}
\usepackage{lineno}
\usepackage{longtable}
\usepackage[utf8]{inputenc}
\usepackage{amsmath}
\usepackage{amsthm}
\usepackage[bitstream-charter]{mathdesign}
\usepackage{siunitx}
\usepackage{adjustbox}
\newtheorem{theorem}{Theorem}
\usepackage{xcolor}
\usepackage{graphicx}
\usepackage[margin=1in]{geometry}
\usepackage{float}
\usepackage{verbatimbox}
\usepackage{booktabs}
\usepackage{algorithm}
\usepackage{algorithmic}

\usepackage[normalem]{ulem}

\usepackage{authblk}
\usepackage{lipsum}
\newcommand{\dis }{\displaystyle}
\usepackage{caption} 
\usepackage{subcaption} 
\usepackage[T1]{fontenc}
\theoremstyle{definition}

\renewcommand{\arraystretch}{1.5}

\makeatletter
\def\thm@space@setup{\thm@preskip=1.2\parskip \thm@postskip=0pt}
\makeatother
\usepackage{tikz}
\usetikzlibrary{arrows}
\usepackage[toc,page]{appendix}
\usepackage{xurl} 
\usepackage[authoryear, round]{natbib}

\usepackage{hyperref}\hypersetup{
	colorlinks=true,
	citecolor=blue,
	linkcolor=red,
	filecolor=blue,
	urlcolor=blue,
}

\begin{document}
	\pagenumbering{gobble} % Disable page numbering
	
	\section*{Graphical Abstract}
	
	\begin{figure}[H]
		\centering
		\includegraphics[width=\textwidth]{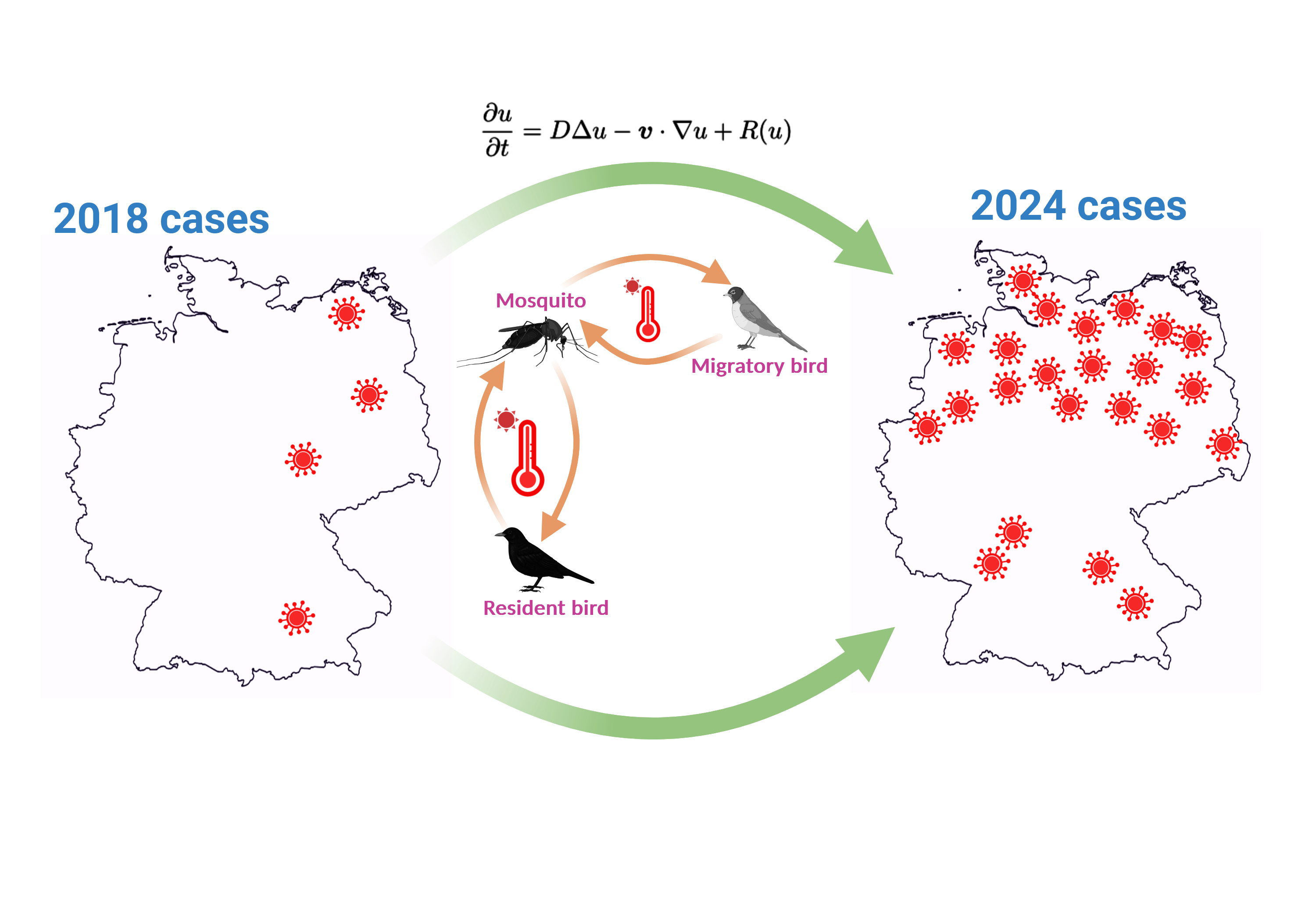}
		\label{f0}
	\end{figure}
	\vspace{0.2cm}
	\begin{center}
		%\small{\textit{Created in BioRender. Arbovirologie, A. (2025) https://BioRender.com/r9wek95. The map was downloaded as a shapefile from (https://gadm.org, version 4.1). Content not licensed under the Creative Commons Attribution (CC BY) license.}}
	\end{center}

	\newpage
	\section*{Highlights}
	
	\vspace{1em}
	\begin{itemize}
		\item A partial differential equation model for the spatial spread of West Nile virus (WNV) in Germany was developed.
		\item The random movement of mosquitoes and birds is defined by the diffusion process.
		\item The directed movement of migratory birds is defined by the advection process.
		\item Temperature-dependent mosquito biting, latency, and mortality rates are considered.
		\item Spreading speeds are estimated from the observed data.
		\item The model is compared to the real-world data of WNV cases in birds, horses and humans.
		\item The model proves to be a promising early warning tool for WNV spreading patterns in Germany and can be adapted for any geographic region.
	\end{itemize}

	\newpage
	
	\clearpage % Start a new page
	\pagenumbering{arabic} % Start page numbering from 1
	\setcounter{page}{1} % Set the page number to 1

	\title{Modeling the impact of temperature and bird migration on the spread of West Nile virus}

	\author{
		Pride Duve\thanks{Corresponding author: \texttt{pride.duve@bnitm.de}}, \;Felix Gregor Sauer, and Renke Lühken \\
		\small	Bernhard Nocht Institute for Tropical Medicine, Hamburg, Germany
	}

	\date{}
	\maketitle

	\begin{abstract}
		
		West Nile virus (WNV) is a climate-sensitive mosquito-borne arbovirus circulating between mosquitoes of the genus \textit{Culex} and birds, with a potential spillover to humans and other mammals. Recent trends in climatic change, characterized by early and/or prolonged summer seasons, increased temperatures, and above-average rainfall, probably facilitated the spread of WNV in Europe. In this work, we formulate a spatial WNV model consisting of a system of parabolic partial differential equations (PDEs), using the concept of diffusion and advection in combination with temperature-dependent parameters, i.e. mosquito biting rate, extrinsic incubation rate, and mortality rate. Diffusion represents the random movement of both mosquitoes and hosts across space, while advection captures the directed movement of migratory birds. The model is first studied mathematically, and we show that it has non-negative, unique, and bounded solutions in time and space. Numerical simulations of the PDE model are performed using temperature data for Germany (2019 - 2024). Results obtained from the simulation showed a high agreement with the reported WNV cases among birds and equids in Germany. The observed spreading patterns from the year 2018 to 2022 and the year 2024 were mainly driven by temperature in combination with diffusion processes of hosts and vectors. Only during the year 2023, the additional inclusion of advection for migratory birds was important to correctly predict new circulation areas in Germany.
		
		\vspace{2em}
		
		\noindent \textbf{Keywords:} West Nile virus, diffusion, advection, spread, mosquitoes, migratory birds, resident birds, hotspots
	\end{abstract}
	
	%\linenumbers
	
	\section{Introduction}
	
	West Nile virus (WNV) is an arbovirus of the Japanese encephalitis antigenic complex of the \textit{Flaviviridae} family. It was first detected in a febrile patient from the West Nile district of Northern Uganda in 1937 \citep{Smithburn1940}, but is now considered the most geographically widespread arbovirus in the world \citep{Ciota2017}. The virus circulates in an enzootic cycle between \textit{Culex} mosquitoes and birds, with regular spillovers to humans, equids, and other mammals \citep{Calisher1989, Hublek1999, Kramer2008}. Birds are considered amplifying hosts as they develop sufficient levels of viremia, allowing for the infection of mosquitoes, while humans, equids, and other mammals are considered dead-end hosts \citep{Rappole2003}. Most humans infected with WNV are asymptomatic, but few may develop fever, headache, rash, among others \citep{Petersen2013}, including the risk of a fatal outcome \citep{Campbell2002}. Infected equids often exhibit ataxia, lethargy, and weakness as signs of the infection \citep{Salazar2004}. Multiple theories have been discussed regarding the introduction of new WNV cases in a previously unaffected region. There is evidence that migratory birds play an important role in the spread of the pathogen \citep{Rappole2003}. For example, birds captured during their seasonal migration have tested positive for WNV in Israel \citep{ernek1977arboviruses, Malkinson2002, NIR1967, watson1972ectoparasite}. Migratory birds likely spread the virus to stopover sites in Europe as they travel from the south to the north \citep{hannoun1972role, Hublek1999}. In contrast, considering the speed and duration of flights of migratory birds in combination with the duration of viremia, other researchers argued that the short-distance movement of non-migratory birds better explains the spread pattern of WNV \citep{Rappole2003}.
	
	Several mathematical models for WNV have been formulated. Many models emphasize the role of climatic conditions on the dynamics of WNV, as temperature plays a critical role in the life-history traits of mosquitoes and virus replication. Thereby, temperature-dependent functions are used to model mosquito and virus development rates \citep{Heidecke2024, Laperriere2011, Mbaoma2024}. These are included in ordinary differential equations (ODEs) to analyze the dynamics of WNV. However, local transmission dynamics are generally modeled within the vector and host populations without considering the geographic spread. \cite{Fesce2022} estimated the basic reproductive number of the WNV infection and used it to quantify the infection spread in mosquitoes and birds, to classify areas with temperatures suitable for the outbreak of WNV. Other studies focus on the control of WNV, e.g., \cite{BOWMAN2005} simulated mosquito reduction methods and personal protective measures. A potential spill-over from birds to equids and humans was modeled by \cite{Laperriere2011} using temperature-driven parameters and providing values for parameters that govern the developments and transitions from one compartment to another. Other models analyzed differences in seasonal trends using sinusoidal (trigonometric) functions to represent seasonal variations of the WNV disease dynamics \citep{abdelrazec2015dynamics, CruzPacheco2009, Moschini2017}. Such models usually only show seasons when mosquitoes are most likely to be active, without factoring in climatic effects.
	
	Different modeling frameworks also include the role of migratory birds in the transmission of WNV. \cite{DBergsman2016} formulated an ODE model that investigates the interplay of WNV in both resident and migratory hosts. The model considers the yearly seasonal outbreaks that depend primarily on the number of susceptible migrant birds entering the local population each season, demonstrating that the early growth rates of seasonal outbreaks are more influenced by the migratory than the resident bird population. For Germany, just recently, \cite{Mbaoma2024} developed an ODE model for WNV by incorporating the complete life cycle of \textit{Cx. pipiens} as a vector; resident birds, migratory birds, and humans as hosts. The study concluded that short-distance migratory birds such as the hooded crow were identified as important hosts in maintaining the enzootic transmission cycle of WNV in Germany. These ecological patterns for bird migration provide the fundamental biological basis for incorporating migratory birds into our model. This is mainly because the bird movement patterns may determine where the virus can be introduced and how fast it can spread during the transmission seasons.
	
	Although ODE-based models allow us to understand the transmission dynamics of WNV, they do not generally consider the spatial movement of hosts and vectors across different regions without incorporating additional types of models, like multi-patch models or distance dispersal kernels. For example, \cite{Bhowmick2023} used an ODE-based modeling approach in combination with different distance dispersal kernels to describe the spreading process of WNV through migratory birds in Germany. As a promising alternative, researchers have incorporated the effects of spatial variations and disease spread across different geographical regions in ODE-based models, using partial differential equation (PDE) models. \cite{Lin2017} developed a reaction-diffusion PDE model for the spatial spread of WNV in mosquitoes and birds in North America using free-boundary conditions. Another study of the WNV propagation was performed by \cite{Maidana2009}, where they proposed a spatial model for analyzing the WNV spread across the USA, calculating the traveling wave speed. The model considers a system of PDEs that describe the movement of mosquito and avian populations using the diffusion and advection processes. Results from this study show that mosquito movements do not play a major role in the spread, while bird advection is an important factor for lower mosquito biting rates. 
	
	\cite{MAIDANA2008, MAIDANA2013} also contributed a lot to the spatial modeling of WNV using PDE models, and their framework is the basis of this study. Our model differs from previous WNV PDE models available in the literature, as we consider two different hosts, i.e., resident and migratory birds, allowing us to capture the changes in the transmission patterns of WNV. Furthermore, we simulate the model in Germany, accounting for spatial heterogeneity, identify hotspots, and give theoretical insights into the contribution of migratory birds to WNV spread. In Germany, WNV was first detected in August 2018 in birds from the Southeastern region of Berlin \citep{Ziegler2019}. Since then, WNV has been established in Germany and has spread. Germany lies along migration routes of birds between Europe, the Mediterranean, and Africa \citep{Musitelli2018}. A study conducted by \cite{Nussbaumer2021} indicates that most migratory birds arrive in the northeastern region of Germany during spring and leave the country in early autumn via the southwestern direction. First WNV strains confirmed in Germany were closely related to strains previously detected in Austria and the Czech Republic \citep{Ziegler2019}, indicating a potential route of introduction, and their ability to overwinter in \textit{Cx. pipiens} \citep{Kampen2021} allows the long-term persistence of the virus.

	\section{Model formulation}
	
	Building on ODE-based models developed by \cite{Laperriere2011} and \cite{Mbaoma2024}, we formulate a model for the spatial spread of WNV in Germany between 2019 and 2024. Our spatial model consists of a system of semi-linear parabolic PDEs governed by diffusion, advection, and reaction processes, together with temperature-dependent mosquito biting, latency, and mortality rates. In addition, our model includes the population of mosquitoes, resident and migratory birds. The population of mosquitoes at time $t$ and geographic location $\displaystyle \boldsymbol{x}= (x,y),$ is divided into three compartments: susceptible $\displaystyle S_V(t,\boldsymbol{x}),$ exposed $\displaystyle E_V(t,\boldsymbol{x})$ and infectious $\displaystyle I_V(t,\boldsymbol{x}),$ with the total population given by $\dis N_V(t,\boldsymbol{x})=S_V(t,\boldsymbol{x})+E_V(t,\boldsymbol{x})+I_V(t,\boldsymbol{x}).$ Susceptible mosquitoes get infected by biting an infectious resident bird $I_{rB}$ and/or migratory bird $I_{mB},$ at a temperature-dependent mosquito biting rate
	\begin{equation}\label{hz}
		\beta(T,\boldsymbol{x})=\dis \frac{0.344}{\dis 1+1.231 e^{\dis -0.184(T-20)}},
	\end{equation}
	and a probability of acquiring the infection of $\dis p_{rV}$ and $\dis p_{mV},$ from resident and migratory birds respectively. The force of infection on mosquitoes is therefore given by 
	\begin{equation*}
		\lambda_{BV}(T,\boldsymbol{x})= \lambda_{r_V}+\lambda_{m_V} = \frac{p_{rV}\beta(T,\boldsymbol{x})I_{rB}}{N_{rB}}+\frac{p_{mM}\beta(T,\boldsymbol{x})I_{mB}}{N_{mB}},
	\end{equation*}
	where $N_{rB}$ and $N_{mB}$ are the total resident and migratory bird populations. A successful contact between infectious birds and susceptible mosquitoes at location $\boldsymbol{x}$ and time $t$ sends susceptible mosquitoes to the exposed class. The latency rate on mosquitoes was fitted by \cite{Heidecke2024}, and is given by 
	\begin{equation}\label{hz1}
		\dis	\gamma_V(T,\boldsymbol{x})= \dis \frac{1}{e^{\dis - 0.09T+5.36}}.
	\end{equation}
	
	and from the exposed stage, mosquitoes progress to the infectious class $(I_V).$ Mosquitoes die at temperature-dependent natural mortality rate of $\mu_V(T,\boldsymbol{x}),$ where 
	\begin{equation}\label{hz2}
		\mu_V(T,\boldsymbol{x}) =\frac{0.0025T^2-0.094T+1.0257}{10}
	\end{equation}
	
	defined by \cite{Laperriere2011}. Spatial and temporal plots of the thermal functions (equations \ref{hz}, \ref{hz1} and, \ref{hz2}) are shown in Appendix S1. The host population consists of amplification hosts: resident birds and migratory birds. The population of hosts is divided into susceptible $S_{j},$ exposed $E_j$, infected $I_{j},$ recovered $R_{j},$ and dead $D_{j}$ hosts, with $\displaystyle j=rB, mB,$ and for easier readability, we denote resident birds by subscript $(rB),$ and migratory birds by $(mB).$ \cite{Laperriere2011} modeled the recruitment of susceptible birds using a gamma distribution, accounting for the observed seasonal birth rate of the American crow, which was the bird species observed to be mainly affected by the WNV in the Minneapolis metropolitan area. To reduce the non-linearity of our model, which has an impact on the numerical solution of the PDE model, we compared our ODE simulation using a recruitment rate of $\dis \Lambda =\mu_{rB} N_{rB},$ but found no significant differences in the numerical results of our ODE version and theirs. Thus, our resident and migratory birds are recruited into the susceptible state at a rate of $\dis \mu_{rB}N_{rB}$ and $\dis \mu_{mB}N_{mB}$ respectively. Susceptible resident and migratory birds get infected by interacting with infectious mosquitoes, with a biting rate $\dis \beta(T,\boldsymbol{x})$ defined in Equation (\ref{hz}), and a probability of acquiring the infection of $p_{Vr}$ and $p_{Vm}$ respectively. In mosquito-borne disease models, the mosquito to bird ratio is very important as it determines the per bird biting pressure \citep{Bhowmick2020, Laperriere2011, Mbaoma2024}. Moreover, \citet{Vogels2017} argue that mosquito abundance is constrained by the availability of host blood meals, and bird density directly influences mosquito feeding success. As a result, the number of bites received by an individual bird depends primarily on the ratio of mosquitoes to birds, rather than on mosquito or bird density alone. Because of this, we introduce the mosquito to bird ratio parameters $\phi_B=(\phi_{rB},\phi_{mB})$ for residence and migratory birds, respectively. In WNV models, \citet{Bhowmick2020} simulated their ODE model with a mosquito to bird ratio ranging from 10-30, while \citet{Laperriere2011} used a ratio of 30. Another study by \cite{deFreitasCosta2024}, which aimed at investigating the conditions necessary for the circulation of WNV in Germany and the Netherlands, found that larger proportions of susceptible birds and higher vector-host ratios were required. The study showed that WNV circulation in the Netherlands was rarely sustained when vector-host ratios fell below approximately 1:100.
	
	The force of infection on resident birds is given by
	
	\begin{equation}
		\lambda_{Vr}(T,\boldsymbol{x}) = \frac{\phi_{rB}p_{Vr}\beta(T,\boldsymbol{x})I_V(t,\boldsymbol{x})}{N_V},
	\end{equation}
	while that of migratory birds is 
	\begin{equation}
		\lambda_{Vm}(T,\boldsymbol{x}) = \frac{\phi_{mB}p_{Vm}\beta(T,\boldsymbol{x})I_V(t,\boldsymbol{x})}{N_V}.
	\end{equation}
	
	The latency rates on resident and migratory birds are given by $\gamma_{rB}$ and $\gamma_{mB}$ respectively, while their recovery rates are given by $\alpha_{rB},$ and $\alpha_{mB}.$ Proportions $1-\nu_{rB}$ of resident birds and $1-\nu_{mB}$ of migratory birds recover from the infection, while the proportions $\nu_{rB}$ and $\nu_{mB}$ die at the rates $\alpha_{rB},$ and $\alpha_{mB}$ respectively. Resident and migratory bird species can also die naturally at rates $\mu_{rB}$ and $\mu_{mB}.$ 
	
	\subsection{Diffusion and advection processes}
	
	Diffusion refers to the process by which matter is transported from one part of a system to another as a result of random molecular motions \citep{crank1979mathematics}. The concept of diffusion dates back to the 19th century, when \cite{Fick1855} recognized the transfer of heat by conduction due to random molecular motions and quantified this process by adopting the mathematical equation of heat conduction previously derived by \cite{baron1822theorie}. To date, the diffusion process has gained significant attention, with mathematical biology being one of the application areas. In ecological modeling, the diffusion process is used to capture the random movement of species from regions of higher concentration to regions of lower concentration \citep{murray2007mathematical}, which has an impact on their spatial distribution patterns. On the other hand, the advection process refers to the directional flow of a substance by a known velocity vector field \citep{boyd2001chebyshev}. Diffusion and advection processes combined with the ODE terms form an advection-diffusion-reaction PDE system that models how a phenomenon spreads in both time and space. Reaction terms are responsible for the local transmission at each location, while diffusion and advection are responsible for the spatial movement of species, including infectious specimens. 
	
In this study, we define diffusion for mosquitoes and both resident and migratory birds, as they locally move in search of resting/breeding sites, food, and other needs. The advection process is defined for migratory birds only, as we can account for their seasonal directional migration patterns. In contrast, diffusion is denoted by the Laplacian operator $\dis (\Delta),$ with diffusion coefficients for mosquitoes, resident, and migratory birds denoted by $\dis D_1, D_2, D_3$ in $\dis \text{km}^2/\text{day}$ respectively. Moreover, we choose the diffusion coefficients in a way that the diffusion of mosquitoes is less than that of resident and migratory birds, as mosquitoes are considered insects with a small radius of action \citep{HARRINGTON2005, Maidana2009, Verdonschot2014}. A lower diffusion coefficient is assigned to resident birds compared to migratory birds to highlight the greater mobility of migratory birds. To improve interpretability of diffusion coefficients, we define the root mean squared displacement (RMSD). The RMSD describes the dispersal distance species move from the initial source under diffusion alone after a time $\tau$ \citep{Einstein1905, metelmann2019development}, and it is calculated using the formula:
	$$\text{RMSD}=\sqrt{\dis 2 p D\tau},$$ where $p=2$ is the spatial dimension, $D$ is the diffusion coefficient and $\tau$ is the time step. We denote the advection term by the gradient operator $\dis (\nabla),$ with a velocity $\dis \boldsymbol{A}=[v_{x},v_{y}]$ \text{km}/day, in a two-dimensional plane. The schematic diagram of our PDE model and the governing system of equations are shown in Figure (\ref{f1}) and system (\ref{sys2a}).
	\begin{figure}[H]
		\centering
		\includegraphics[width=1\textwidth]{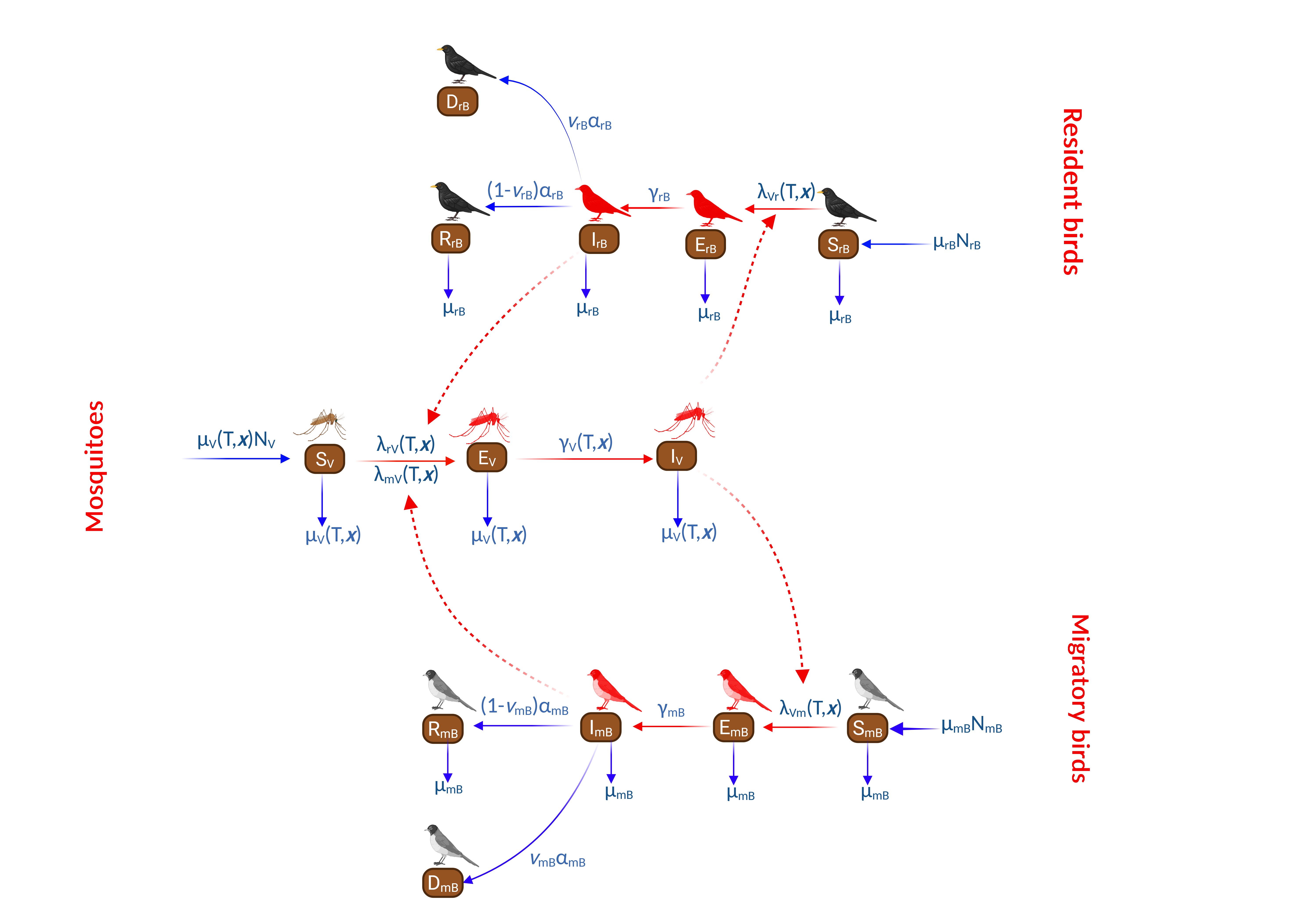}
		\caption{\small  Flow chart diagram describing the transmission dynamics of WNV in Germany. Solid arrows represent transitions between compartments, while dotted lines represent interaction between different states. Image created in BioRender. Arbovirologie, A. (2025) https://BioRender.com/8a0gfj1. Content not licensed under the Creative Commons Attribution (CC BY) license.}
		\label{f1}
	\end{figure}

	\begin{equation}\label{sys2a}   
		\begin{split}\shortstack[l]{Mosquito\\population:}\;& \begin{cases}
				\displaystyle	\frac{\partial S_V}{\partial t} & =\quad\displaystyle D_{1}\Delta S_V\quad+\quad \mu_V(T,\boldsymbol{x})N_V \quad-\quad [\lambda_{rV}(T,\boldsymbol{x})+\lambda{mV}(T,\boldsymbol{x})+\mu_V(T,\boldsymbol{x})]S_V,\\	
				\\\displaystyle	\frac{\partial E_V}{\partial t} & =\quad\displaystyle D_{1}\Delta E_V\quad+\quad [\lambda_{rV}(T,\boldsymbol{x})+\lambda{mV}(T,\boldsymbol{x})]S_V\quad-\quad\left[\gamma_V(T,\boldsymbol{x})+\mu_V(T,\boldsymbol{x})\right]E_V,\\	
				\\\displaystyle	\frac{\partial I_V}{\partial t} & =\quad\displaystyle D_{1}\Delta I_V\quad+\quad \gamma_V(T,\boldsymbol{x})E_V\quad-\quad\mu_V(T,\boldsymbol{x})I_V,\end{cases}\\
			\\\text{Resident birds:} & \begin{cases}
				\displaystyle	\frac{\partial S_{rB}}{\partial t} & =\quad\displaystyle D_{2}\Delta S_{rB} \quad+\quad \mu_{rB}N_{rB}\quad-\quad\left[\lambda_{Vr}(T,\boldsymbol{x})+\mu_{rB}\right]S_{rB},\\
				\\\displaystyle	\frac{\partial E_{rB}}{\partial t} & =\quad\displaystyle D_{2}\Delta E_{rB} \quad+\quad \lambda_{Vr}(T,\boldsymbol{x})S_{rB}\quad-\quad[\gamma_{rB}+\mu_{rB}]E_{rB},\\
				\\\displaystyle	\frac{\partial I_{rB}}{\partial t} & =\quad\displaystyle D_{2}\Delta I_{rB} \quad+\quad \gamma_{rB}E_{rB}\quad-\quad[\alpha_{rB}+\mu_{rB}]I_{rB},\\
				\\\displaystyle	\frac{\partial R_{rB}}{\partial t} & =\quad\displaystyle D_{2}\Delta R_{rB}\quad+\quad (1-\nu_{rB})\alpha_{rB} I_{rB}\quad-\quad\mu_{rB}R_{rB},\\
				\\\displaystyle	\frac{\partial D_{rB}}{\partial t} & =\displaystyle\quad \alpha_{rB}\nu_{rB} I_{rB},\\\end{cases}\\
			\\\text{Migratory birds:} & \begin{cases}
				\displaystyle	\frac{\partial S_{mB}}{\partial t} & =\quad\displaystyle D_{3}\Delta S_{mB}  -\boldsymbol{A}\cdot\nabla S_{mB}\quad+ \quad  \mu_{mB}N_{mB}\quad-\quad\left[\lambda_{Vm}(T,\boldsymbol{x})+\mu_{mB}\right]S_{mB},\\
				\\\displaystyle	\frac{\partial E_{mB}}{\partial t} & =\quad\displaystyle D_{3}\Delta E_{mB}  - \boldsymbol{A}\cdot\nabla E_{mB}\quad+ \quad \lambda_{Vm}(T,\boldsymbol{x})S_{mB}\quad-\quad[\gamma_{mB}+\mu_{mB}]E_{mB},\\
				\\\displaystyle	\frac{\partial I_{mB}}{\partial t} & =\quad\displaystyle D_{3}\Delta I_{mB}  - \boldsymbol{A}\cdot\nabla I_{mB}\quad+ \quad \gamma_{mB}E_{mB}\quad-\quad[\alpha_{mB}+\mu_{mB}]I_{mB},\\
				\\\displaystyle	\frac{\partial R_{mB}}{\partial t} & =\quad\displaystyle D_{3}\Delta R_{mB}  - \boldsymbol{A}\cdot\nabla R_{mB}\quad+ \quad(1-\nu_{mB})\alpha_{mB}  I_{mB}\quad-\quad\mu_{mB}R_{mB},\\
				\\\displaystyle	\frac{\partial D_{mB}}{\partial t} & =\quad\displaystyle  \nu_{mB}\alpha_{mB} I_{mB}.\\\end{cases}
		\end{split}
	\end{equation}

	The total populations of mosquitoes, resident birds and migratory birds are given by $\displaystyle 
	N_{V}=S_V+E_V+I_V,\;  N_{rB}=S_{rB}+E_{rB}+I_{rB}+R_{rB},\; N_{mB}=S_{mB}+E_{mB}+I_{mB}+R_{mB},$ while the initial conditions are given by:
	
	$S_V(0,\boldsymbol{x})=\psi_1(\boldsymbol{x}),$ $E_V(0)=\psi_2(\boldsymbol{x}),$ $I_V(0)=\psi_3(\boldsymbol{x}),$ $S_{rB}(0,\boldsymbol{x})=\psi_4(\boldsymbol{x}),$  $E_{rB}(0,\boldsymbol{x})=\psi_5(\boldsymbol{x}),$ $I_{rB}(0,\boldsymbol{x})=\psi_6(\boldsymbol{x}),$ $R_{rB}(0,\boldsymbol{x})=\psi_7(\boldsymbol{x}),$ $D_{rB}(0)\geq 0,$
	$S_{mB}(0,\boldsymbol{x})=\psi_8(\boldsymbol{x}),$ $E_{mB}(0,\boldsymbol{x})=\psi_{9}(\boldsymbol{x}),$ $I_{mB}(0,\boldsymbol{x})=\psi_{10}(\boldsymbol{x}),$ $R_{mB}(0,\boldsymbol{x})=\psi_{11}(\boldsymbol{x}),$ $D_{mB}(0)\geq 0,$ where, 
	
	$$\displaystyle \boldsymbol{x}=(x,y)\in\Omega, \quad \Delta =\frac{\partial^2(\cdot)}{\partial x^2}+\frac{\partial^2(\cdot)}{\partial y^2},\quad \nabla =\left[ \frac{\partial (\cdot)}{\partial x},\frac{\partial (\cdot)}{\partial y}\right].$$
	
	We assume the system moves in a region $\displaystyle \Omega\subset \mathbb{R}^2,$ with a smooth boundary $\partial\Omega$ according to Fick's law \citep{Fick1855}, so that in the initial conditions, $\boldsymbol{x}\in\Omega,$ where $\psi_i\in C^2(\Omega)\cap C(\bar{\Omega})$ and subject to the homogeneous Neumann boundary conditions:
	
	$$\frac{\partial S_V}{\partial n}=\frac{\partial E_V}{\partial n}=\frac{\partial I_V}{\partial n}=\frac{\partial S_{rB}}{\partial n}=\frac{\partial E_{rB}}{\partial n}=\frac{\partial I_{rB}}{\partial n}=\frac{\partial R_{rB}}{\partial n}=\frac{\partial S_{mB}}{\partial n}=\frac{\partial E_{mB}}{\partial n}=\frac{\partial I_{mB}}{\partial n}=\frac{\partial R_{mB}}{\partial n}=0,$$
	
	$\displaystyle \text{for}\; \boldsymbol{x}\in\partial\Omega,\;t>0,$ and $n$ is the unit outer normal to $\Omega.$ The state variables and initial conditions are summarized in Table (S1) of the supplementary file, while model parameters are shown in Table (\ref{t02}).

	\begin{table}[H]
		\caption{\small Definition of parameters used in system (\ref{sys2a}).}
		\renewcommand{\arraystretch}{1.4}
		\small
		\centering
		\begin{adjustbox}{width=1.1\textwidth}
			\begin{tabular}{llll}
				\toprule 
				\textbf{Parameter} & \textbf{Definition}&\textbf{Value}& \textbf{Source} \\
				\midrule 		
				$\mu_{rB}(\boldsymbol{x})$ & natural mortality rate of resident birds&0.0005&\citep{Mbaoma2024}\\
				$\mu_{mB}(\boldsymbol{x})$ & natural mortality rate of migratory birds&0.00023&\citep{Mbaoma2024}\\
				
				$\gamma_{V}(T,\boldsymbol{x})$&latency rate on mosquitoes&Eqn (\ref{hz1})&\citep{Heidecke2024}\\
				$\gamma_{rB}(\boldsymbol{x})$&latency rate on resident birds&0.196&\citep{Mbaoma2024}\\
				$\gamma_{mB}(\boldsymbol{x})$&latency rate on migratory birds&0.285&\citep{Mbaoma2024}\\
				
				$\alpha_{rB}(\boldsymbol{x})$&removal rate of infectious resident birds&0.867&\citep{Mbaoma2024}\\
				$\alpha_{mB}(\boldsymbol{x})$&removal rate of infectious migratory birds&0.4&\citep{Mbaoma2024}\\
				
				$\nu_{rB}(\boldsymbol{x})$&proportion of dead resident birds&0.655&\citep{Mbaoma2024}\\
				$\nu_{mB}(\boldsymbol{x})$&proportion of dead migratory birds&0.103&\citep{Mbaoma2024}\\
				
				$p_{Vr}(\boldsymbol{x})$ &probability of a new resident bird infection&0.97&\citep{Mbaoma2024}\\ 
				$p_{Vm}(\boldsymbol{x})$ &probability of a new migratory bird infection&0.9&\citep{Mbaoma2024}\\ 
				
				$p_{rV}(\boldsymbol{x})$ &probability of a new mosquito infection from a resident bird&0.4&\citep{Mbaoma2024}\\ 
				$p_{mV}(\boldsymbol{x})$ &probability of a new mosquito infection from a migratory bird&0.7&\citep{Mbaoma2024}\\ 
				$\beta(T,\boldsymbol{x})$ & mosquito biting rate&Eqn (\ref{hz})&\citep{rubel2008}\\ 
				$\mu_V(T,\boldsymbol{x})$ & natural mortality rate of mosquitoes&Eqn (\ref{hz2})&\citep{Laperriere2011}\\ 
				$\phi_{rB}(\boldsymbol{x})$ & mosquito to resident bird ratio&15&inferred from the data\\
				$\phi_{mB}(\boldsymbol{x})$ & mosquito to migratory bird ratio&15&inferred from the data\\
				$D_1(\boldsymbol{x})$ & diffusion coefficient of mosquitoes&Table (\ref{t03})&inferred from the data\\
				$D_2(\boldsymbol{x})$ & diffusion coefficient of resident birds&Table (\ref{t03})&inferred from the data\\
				$D_3(\boldsymbol{x})$ & diffusion coefficient of migratory birds&Table (\ref{t03})&inferred from the data\\
				$\boldsymbol{A}(t,\boldsymbol{x})=[v_x;v_y]$ & advection vector for migratory birds&[1;1] \& [-1;-1]&Estimated section (\ref{rs})\\
				\bottomrule
			\end{tabular}
		\end{adjustbox}
		\label{t02}
	\end{table}

	\section{Mathematical analysis of the model}
	
	\begin{theorem}
		System (\ref{sys2a}) has non-negative and unique solutions that are bounded in $\displaystyle [0,\infty).$
	\end{theorem}

	\begin{proof}
		
		We write system (\ref{sys2a}) together with the corresponding initial conditions, in the Banach space of continuous functions that include the boundary, $\displaystyle \mathcal{B}=C(\bar{\Omega}),$ as follows \citep{abboubakar2023reaction}:
		
		\begin{equation}\label{er}
			\begin{cases}
				\displaystyle \frac{\partial X(t,\boldsymbol{x})}{\partial t} =&\displaystyle \mathcal{L}X(t,\boldsymbol{x}) + f(X(t,\boldsymbol{x})),\; t>0,\\
				\displaystyle X(t,0) =&\displaystyle X_0\geq \boldmath{0}_{\mathbb{R}^{11}},
		\end{cases}\end{equation}
		
		where 
		
		$$X(t,\boldsymbol{x})=\begin{bmatrix}
			S_V(t,\boldsymbol{x})\\
			E_V(t,\boldsymbol{x})\\
			I_V(t,\boldsymbol{x})\\
			S_{rB}(t,\boldsymbol{x})\\
			E_{rB}(t,\boldsymbol{x})\\
			I_{rB}(t,\boldsymbol{x})\\
			R_{rB}(t,\boldsymbol{x})\\
			S_{mB}(t,\boldsymbol{x})\\
			E_{mB}(t,\boldsymbol{x})\\
			I_{mB}(t,\boldsymbol{x})\\
			R_{mB}(t,\boldsymbol{x})\\
		\end{bmatrix},\; X(0,\boldsymbol{x})=\begin{bmatrix}
			S_V(0,\boldsymbol{x})\\
			E_V(0,\boldsymbol{x})\\
			I_V(0,\boldsymbol{x})\\
			S_{rB}(0,\boldsymbol{x})\\
			E_{rB}(0,\boldsymbol{x})\\
			I_{rB}(0,\boldsymbol{x})\\
			R_{rB}(0,\boldsymbol{x})\\
			S_{mB}(0,\boldsymbol{x})\\
			E_{mB}(0,\boldsymbol{x})\\
			I_{mB}(0,\boldsymbol{x})\\
			R_{mB}(0,\boldsymbol{x})\\
		\end{bmatrix}, \; \mathcal{L}X(t,\boldsymbol{x}) = \begin{bmatrix}
			D_{1}\Delta S_V(t,\boldsymbol{x})\\
			D_{1}\Delta E_V(t,\boldsymbol{x})\\
			D_{1}\Delta I_V(t,\boldsymbol{x})\\
			D_{2}\Delta S_{rB}(t,\boldsymbol{x})\\
			D_{2}\Delta E_{rB}(t,\boldsymbol{x})\\
			D_{2}\Delta I_{rB}(t,\boldsymbol{x})\\
			D_{2}\Delta R_{rB}(t,\boldsymbol{x})\\
			D_{3}\Delta S_{mB}(t,\boldsymbol{x})-\boldsymbol{A}\cdot\nabla S_{mB}(t,\boldsymbol{x})\\
			D_{3}\Delta E_{mB}(t,\boldsymbol{x})-\boldsymbol{A}\cdot\nabla E_{mB}(t,\boldsymbol{x})\\
			D_{3}\Delta I_{mB}(t,\boldsymbol{x})-\boldsymbol{A}\cdot\nabla I_{mB}(t,\boldsymbol{x})\\
			D_{3}\Delta R_{mB}(t,\boldsymbol{x})-\boldsymbol{A}\cdot\nabla R_{mB}(t,\boldsymbol{x})\\
		\end{bmatrix},$$
		
		and	
		
		\begin{eqnarray} \label{g1} 
			f =	\begin{cases}
				f_1=& \mu_VN_V -[\lambda_{rV}(T)+\lambda{mV}(T)+\mu_V]S_V,\nonumber\\			
				f_2=& \left[\lambda_{rV}(T)+\lambda{mV}(T)\right]S_V-(\gamma_V+\mu_V)E_V,\nonumber\\
				f_3=& \gamma_VE_V-\mu_VI_V,\nonumber\\
				f_4=& \mu_{rB}N_{rB}-\left[\lambda_{Vr}(T)+\mu_{rB}\right]S_{rB},\nonumber\\
				f_5=& \lambda_{Vr}(T)S_{rB}-[\gamma_{rB}+\mu_{rB}]E_{rB},\nonumber\\
				f_6=& \gamma_{rB}E_{rB}-[\alpha_{rB}+\mu_{rB}]I_{rB},\nonumber\\
				f_7=& (1-\nu_{rB})\alpha_{rB} I_{rB}-\mu_{rB}R_{rB},\nonumber\\
				f_8=& \mu_{mB}N_{mB}-\left[\lambda_{Vm}(T)+\mu_{mB}\right]S_{mB},\nonumber\\
				f_9=&\lambda_{Vm}(T)S_{mB}-[\gamma_{mB}+\mu_{mB}]E_{mB},\nonumber\\
				f_{10}=& \gamma_{mB}E_{mB}-[\alpha_{mB}+\mu_{mB}]I_{mB},\nonumber\\
				f_{11}=&(1-\nu_{mB})\alpha_{mB}  I_{mB}-\mu_{mB}R_{mB}.\nonumber\\
			\end{cases}
		\end{eqnarray}

		Next, we show that $f$ is locally Lipschitz in $\mathcal{B}.$ Thus we show that 
		$$\forall\; K\subset \mathcal{B},\; \exists L:\quad \|f(X_1)-f(X_2)\|_{\infty}\leq L\|X_1-X_2\|_{\infty},\;\forall X_1,X_2\in K,$$
		where $L$ is the Lipschitz constant. We observe that 
		$\displaystyle X_1-X_2=$
		$$
		\begin{bmatrix}
			\mu_V\left[N_{V_1}-N_{V_2}\right]-\left[\lambda_{rV_1}S_{V_1}-\lambda_{rV_2}S_{V_2}\right]-\left[\lambda_{mV_1}S_{V_1}-\lambda_{mV_2}S_{V_2}\right]-\mu_V(S_{V_1}-S_{V_2})\\
			\left[\lambda_{rV_1}S_{V_1}-\lambda_{rV_2}S_{V_2}\right]+\left[\lambda_{mV_1}S_{V_1}-\lambda_{mV_2}S_{V_2}\right]-(\gamma_V+\mu_V)[E_{V_1}-E_{V_2}]\\
			\gamma_V[E_{V_1}-E_{V_2}]-\mu_V[I_{V_1}-I_{V_2}]\\
			\mu_{rB}[N_{rB_1}-N_{rB_2}]-(\lambda_{Vr_1}S_{rB_1}-\lambda_{Vr_2}S_{rB_2})-\mu_{rB}[S_{rB_1}-S_{rB_2}]\\
			\lambda_{Vr_1}S_{rB_1}-\lambda_{Vr_2}S_{rB_2}-(\gamma_{rB}+\mu_{rB})[E_{rB_1}-E_{rB_2}]\\
			\gamma_{rB}[E_{rB_1}-E_{rB_2}] - (\alpha_{rB}+\mu_{rB})[I_{rB_1}-I_{rB_2}]\\
			(1-\nu_{rB})\alpha_{rB}[I_{rB_1}-I_{rB_2}]-\mu_{rB}[R_{rB_1}-R_{rB_2}]\\
			\mu_{mB}[N_{mB_1}-N_{mB_2}] - [\lambda_{Vm_1}S_{mB_1}-\lambda_{Vm_2}S_{mB_2}]-\mu_{mB}[S_{mB_1}-S_{mB_2}]\\
			\lambda_{Vm_1}S_{mB_1}-\lambda_{Vm_2}S_{mB_2}-(\gamma_{mB}+\mu_{mB})[E_{mB_1}-E_{mB_2}]\\
			\gamma_{mB}[E_{mB_1}-E_{mB_2}] - (\alpha_{mB}+\mu_{mB})[I_{mB_1}-I_{mB_2}]\\
			(1-\nu_{mB})\alpha_{mB}[I_{mB_1}-I_{mB_2}]-\mu_{mB}[R_{mB_1}-R_{mB_2}]\\	
		\end{bmatrix}.$$

		Simplifying, we obtain $\displaystyle \|f(X_1)-f(X_2)\|_{\infty}=$
		
		\begin{eqnarray}
			&=& \displaystyle\sup_{ x\in\bar{\Omega}}|\mu_V\left[N_{V_1}-N_{V_2}\right]-S_{V_1}[\lambda_{rV_1}-\lambda_{rV_2}]|\nonumber\\
			&\vee&\displaystyle\sup_{ x\in\bar{\Omega}}|-\lambda_{rV_2}[S_{V_1}-S_{V_2}]-S_{V_1}[\lambda_{mV_1}-\lambda_{mV_2}]-\lambda_{mV_2}[S_{V_1}-S_{V_2}]-\mu_V(S_{V_1}-S_{V_2})|\nonumber\\
			&\vee&\displaystyle\sup_{ x\in\bar{\Omega}} |\gamma_V[E_{V_1}-E_{V_2}]-\mu_V[I_{V_1}-I_{V_2}]|  \nonumber\\
			&\vee&\displaystyle\sup_{ x\in\bar{\Omega}} |\mu_{rB}[N_{rB_1}-N_{rB_2}]-\lambda_{Vr_1}[S_{rB_1}-S_{rB_2}]-S_{rB_2}[\lambda_{Vr_1}-\lambda_{Vr_2}]-\mu_{rB}[S_{rB_1}-S_{rB_2}]|  \nonumber\\
			&\vee&\displaystyle\sup_{ x\in\bar{\Omega}} |\lambda_{Vr_1}[S_{rB_1}-S_{rB_2}]+S_{rB_2}[\lambda_{Vr_1}-\lambda_{Vr_2}]-(\gamma_{rB}+\mu_{rB})[E_{rB_1}-E_{rB_2}]| \nonumber\\
			&\vee&\displaystyle\sup_{ x\in\bar{\Omega}} |\gamma_{rB}[E_{rB_1}-E_{rB_2}] - (\alpha_{rB}+\mu_{rB})[I_{rB_1}-I_{rB_2}]| \nonumber\\
			&\vee&\displaystyle\sup_{ x\in\bar{\Omega}} |(1-\nu_{rB})\alpha_{rB}[I_{rB_1}-I_{rB_2}]-\mu_{rB}[R_{rB_1}-R_{rB_2}]|\nonumber\\
			&\vee&\displaystyle\sup_{ x\in\bar{\Omega}} |\mu_{mB}[N_{mB_1}-N_{mB_2}]-\lambda_{Vm_1}[S_{mB_1}-S_{mB_2}]-S_{mB_2}[\lambda_{Vm_1}-\lambda_{Vm_2}]-\mu_{mB}[S_{mB_1}-S_{mB_2}]|  \nonumber\\
			&\vee&\displaystyle\sup_{ x\in\bar{\Omega}} |\lambda_{Vm_1}[S_{mB_1}-S_{mB_2}]+S_{mB_2}[\lambda_{Vm_1}-\lambda_{Vm_2}]-(\gamma_{mB}+\mu_{mB})[E_{mB_1}-E_{mB_2}]| \nonumber\\
			&\vee&\displaystyle\sup_{ x\in\bar{\Omega}} |\gamma_{mB}[E_{mB_1}-E_{mB_2}] - (\alpha_{mB}+\mu_{mB})[I_{mB_1}-I_{mB_2}]| \nonumber\\
			&\vee&\displaystyle\sup_{ x\in\bar{\Omega}} |(1-\nu_{mB})\alpha_{mB}[I_{mB_1}-I_{mB_2}]-\mu_{mB}[R_{mB_1}-R_{mB_2}]|.\nonumber\\
		\end{eqnarray}

		By the triangular inequality, and some simplifications, we arrive at
		
		$$\|f(X_1)-f(X_2)\|_{\infty}\leq \left(\mu_V+\nu_{rB}\alpha_{rB}+\nu_{mB}\alpha_{mB}\right)\|X_1-X_2\|_{\infty},$$
		
		and thus $f$ is locally Lipschitz in $\mathcal{B}.$ By [\citep{Capasso1993}, Theorem B.17], \citep{Mora1983}, and [\citep{Smoller1983}, Theorem 14.4], there exist a local smooth and unique solution of system (\ref{sys2a}) in $\displaystyle \Omega.$ We observe that system (\ref{sys2a}) can be written in the form of system (14.12) in the book \citep{Smoller1983}, together with initial data defined in system (14.13). By [\citep{Smoller1983}, Theorem (14.14)], the solutions of system (\ref{sys2a}) are always positive.	
	\end{proof}

	\subsection{Well-posedness of the model in a feasible region}
	
	\begin{theorem}
		System (\ref{sys2a}) is well-posed mathematically and biologically, and for any 
		
		$\dis \left(S_V(0,x),E_V(0,x),I_V(0,x),S_{rB}(0,x), E_{rB}(0,x), I_{rB}(0,x), R_{rB}(0,x), S_{mB}(0,x), E_{mB}(0,x),\right.$
		
		$\left. I_{mB}(0,x),R_{mB}(0,x)\right)\in \mathbb{X},$ system (\ref{sys2a}) admits a unique positive solution: $\dis \left(S_V(t,x),E_V(t,x),I_V(t,x),\right.$
		
		$\left.S_{rB}(t,x),E_{rB}(t,x), I_{rB}(t,x), R_{rB}(t,x), S_{mB}(t,x),E_{mB}(t,x), I_{mB}(t,x), R_{mB}(t,x)\right)\in \mathbb{X},$ satisfying:
		
		$\dis (S_V(t,x),E_V(t,x),I_V(t,x),S_{rB}(t,x), E_{rB}(t,x), I_{rB}(t,x), R_{rB}(t,x),S_{mB}(t,x),E_{mB}(t,x), I_{mB}(t,x), R_{mB}(t,x) \in C^{1,2}((0,\infty)\times\bar{\Omega}) \times C^{1,2}((0,\infty)\times\bar{\Omega}),\;\text{where}\; \mathbb{X}:=C(\bar{\Omega})\times C(\bar{\Omega}).$
		
		Moreover, there exist another constant $C_1>0$ independent of initial data, such that the solution $\dis (S_V(t,x),E_V(t,x),I_V(t,x),S_{rB}(t,x), E_{rB}(t,x), I_{rB}(t,x), R_{rB}(t,x),S_{mB}(t,x), E_{mB}(t,x), I_{mB}(t,x), R_{mB}(t,x))$ satisfies:
		
		$\|S_V(t,x)\|_{L^{\infty}(\Omega)}+\|E_V(t,x)\|_{L^{\infty}(\Omega)} +\|I_V(t,x)\|_{L^{\infty}(\Omega)} +\|S_{rB}(t,x)\|_{L^{\infty}(\Omega)} +\|E_{rB}(t,x)\|_{L^{\infty}(\Omega)}+\\ \|I_{rB}(t,x)\|_{L^{\infty}(\Omega)}+\|R_{rB}(t,x)\|_{L^{\infty}(\Omega)} +\|S_{mB}(t,x)\|_{L^{\infty}(\Omega)} +\|E_{mB}(t,x)\|_{L^{\infty}(\Omega)} +\|I_{mB}(t,x)\|_{L^{\infty}(\Omega)}+\\\|R_{mB}(t,x)\|_{L^{\infty}(\Omega)} \leq C_1,\;\text{for all}\; t>T_0>0.$
		
	\end{theorem}

	\begin{proof}
		
		We follow the approach presented in \citep{Peng2012, Wang2022}. Using the regularity theory of parabolic PDEs \citep{Pao1993}, system (\ref{sys2a}) admits a unique non-negative classical solution 
		
		$$\dis (S_j(t,x), E_j(t,x), I_j(t,x), R_j(t,x)) \in C^{1,2}((0,T_m)\times\bar{\Omega}) \times C^{1,2}((0,T_m)\times\bar{\Omega}),$$
		
		where $T_m$ represents the maximal existence time of the solution. By the strong maximum principle \citep{Protter1984}, then $S_j(t,x), E_j(t,x), I_j(t,x), R_j(t,x),$ where $\dis j=rB,mB$ are positive in $\dis (0,T_m)\times\bar{\Omega}.$
		
		Summing up the equations and integrating over 
		
		$\dis \Omega,$ we get $\displaystyle \int_{\Omega} \left[S_j(t,x)+E_j(t,x)+I_j(t,x)+R_j(t,x)\right] \;dx=N_{j_0}>0.$
		
		$\dis N_{j_0} = \displaystyle \int_{\Omega} \left[S_j(t,x)+E_j(t,x)+I_j(t,x)+R_j(t,x)\right] \;dx, \implies \dis \|S_j(t,\cdot)\|_{L^1(\Omega)},$ $\dis \|E_V(t,\cdot)\|_{L^1(\Omega)}$ $\dis \|I_V(t,\cdot)\|_{L^1(\Omega)}$ and $\dis \|R_j(t,x)\|_{L^1(\Omega)}$ are bounded for all $\dis 0<t<T_m.$ By the positivity of $S_j(t,\cdot),$ $E_j(t,\cdot),$ $I_j(t,\cdot),$ $R_j(t,\cdot),$ and [\citep{Peng2012}, Lemma (3.1)], with $\sigma=p_0=1,$ we conclude that there exist a positive constant $C_1$ that does not depend on initial data such that the solution $S_j, E_j, I_j, R_j$ satisfies
		$$\|S_j(t,\cdot)\|_{L^{\infty}(\Omega)} +  \|E_j(t,\cdot)\|_{L^{\infty}(\Omega)}+  \|I_j(t,\cdot)\|_{L^{\infty}(\Omega)} +  \|R_j(t,\cdot)\|_{L^{\infty}(\Omega)}\leq C_1,\;\forall t>T_0.$$
		
		Similar arguments can be made for the mosquito population, thus we conclude that system (\ref{sys2a}) is well-posed.
	\end{proof}

	\section{Steady states and the reproductive number}
	
	The PDE system (\ref{sys2a}) admits a non-trivial WNV-free equilibrium point $\mathcal{E}^0$ such that 
	
	$$\left[S_V,E_V,I_V,S_{rB},E_{rB},I_{rB},R_{rB},D_{rB}, S_{mB},E_{mB},I_{mB},R_{mB},D_{mB}\right]$$$$=\left[N_V^0, 0,0,N_{rB}^0,0,0,0,0,N_{mB}^0,0,0,0,0\right].$$ The uniqueness of $\mathcal{E}^0$ can be verified using [\citep{Wang2018}, Lemma 2.1]. The basic reproductive number of the deterministic part of our model is computed following \cite{vandenDriessche2002}, as follows:
	
Let $\mathcal{F},$ be the matrix of new infections, and $\mathcal{V},$ the matrix of transitions, then
	
	$$\mathcal{F}=\renewcommand{\arraystretch}{1.6}\left[\setlength\arraycolsep{1.0pt}
	\begin{array}{c}
		(\lambda_{rV}+\lambda_{Vm})S_V \\
		0\\
		\lambda_{Vr}S_{rB}\\
		0\\
		\lambda_{Vm}S_{mB}\\
		0\\
	\end{array}\right],\quad
	\mathcal{V}=\renewcommand{\arraystretch}{1.4}\left[\setlength\arraycolsep{1.3pt}
	\begin{array}{c}
		(\gamma_V+\mu_V)E_V \\
		-\gamma_VE_V+\mu_VI_V\\
		(\gamma_{rB}+\mu_{rB})E_{rB}\\
		-\gamma_{rB}E_{rB}+(\alpha_{rB}+\mu_{rB})I_{rB}\\
		(\gamma_{mB}+\mu_{mB})E_{mB}\\
		-\gamma_{mB}E_{mB}+(\alpha_{mB}+\mu_{mB})I_{mB}\\
	\end{array}\right].
	$$
	
Partial derivatives evaluated at the WNV-free equilibrium point yield:

	$$F=\left[\renewcommand{\arraystretch}{1.6}
	\begin{array}{cccccccccc}
		0 & 0 & 0 &\dis  \frac{\phi_{rB}\beta  N^0_V p_{rV}}{N_{rB}} & 0 & \dis \frac{\phi_{mB}\beta  p_{mV} N^0_V}{N_{mB}} \\
		0 & 0 & 0 & 0 & 0 & 0  \\
		0 &\dis  \frac{\beta  N^0_{rB} p_{Vr}}{N_V} & 0 & 0 & 0 & 0  \\
		0 & 0 & 0 & 0 & 0 & 0 \\
		0 &\dis  \frac{\beta  N^0_{mB} p_{Vm}}{N_V} & 0 & 0 & 0 & 0 \\
		0 & 0 & 0 & 0 & 0 & 0  \\
	\end{array}
	\right],$$
	
	and

	$$V=\left[\renewcommand{\arraystretch}{1.2}\setlength\arraycolsep{3pt}
	\begin{array}{cccccc}
		-\gamma _V-\mu _V & 0 & 0 & 0 & 0  &0\\
		\gamma _V & -\mu _V & 0 & 0 & 0&0 \\
		0 & 0 & -\gamma _{rB}-\mu _{rB} & 0 & 0 &0\\
		0 & 0 & \gamma _{rB} & -\alpha _{rB}-\mu _{rB} & 0 &0\\
		0 & 0 & 0 & 0 & -\gamma _{mB}-\mu _{mB} &0\\
		0 & 0 & 0 & 0 & \gamma _{mB}&-\alpha _{mB}-\mu _{mB} 
	\end{array}
	\right].$$

	The basic reproductive number $(R_0(\boldsymbol{x}))$ is thus given by the spectral radius of $\displaystyle FV^{-1},$ which gives
	
	\begin{equation}
		R_{0}(\boldsymbol{x}) =\sqrt{	\underbrace{\left(\frac{\phi_{rB}\beta p_{r_V}\gamma_V}{\mu_V(\gamma_V+\mu_V)}\right)\cdot\left(\frac{\beta p_{V_r}\gamma_{rB}}{(\alpha_{rB}+\mu_{r_B})(\gamma_{rB}+\mu_{rB})}\right)}_{\text{Infections due to resident birds}}   +   	\underbrace{\left(\frac{\phi_{mB}\beta p_{m_V}\gamma_V}{\mu_V(\gamma_V+\mu_V)}\right)\cdot\left(\frac{\beta p_{V_m}\gamma_{mB}}{(\alpha_{mB}+\mu_{m_B})(\gamma_{mB}+\mu_{mB})}\right)}_{\text{Infections due to migratory birds}}}.
	\end{equation}
	
	The basic reproductive number can be expressed as $\dis R_0 = \sqrt{R_{0}^{Res}+R_{0}^{Mig}},$ where 
	
	$$\dis 
	R_{0}^{Res} =\left(\frac{\phi_{rB}\beta p_{r_V}\gamma_V}{\mu_V(\gamma_V+\mu_V)}\right)\times\left(\frac{\beta p_{V_r}\gamma_{rB}}{(\alpha_{rB}+\mu_{r_B})(\gamma_{rB}+\mu_{rB})}\right)\quad\text{and}$$
	
	$$R_{0}^{Mig} = \left(\frac{\phi_{mB}\beta p_{m_V}\gamma_V}{\mu_V(\gamma_V+\mu_V)}\right)\times\left(\frac{\beta p_{V_m}\gamma_{mB}}{(\alpha_{mB}+\mu_{m_B})(\gamma_{mB}+\mu_{mB})}\right).$$

	Since all parameters of the model are positive, and the equilibrium point is also positive, then $R_{0}^{Res}$ and $R_{0}^{Mig}$ are non-negative. From this, we observe that 
	
	\begin{equation}\label{rn}
		\sqrt{R_{0}^{Res}+R_{0}^{Mig}}\geq\displaystyle\sqrt{R_{0}^{Res}},
	\end{equation} 
	
with equality when $R_{0}^{Mig}=0,$ which shows the contribution of migratory birds in increasing the amplifying host population. In Figure (\ref{fb}), the basic reproductive number $R_0(\boldsymbol{x})$ is plotted for the entire country and compared to the observed data for the years 2018-2024. $R_0(\boldsymbol{x})$ shows a bigger hotspot in the eastern parts of Germany, with other circulation areas observed in the Southeastern towards the Northern parts, which is in agreement with the observed data. Notably, we also included the 2025 case data to evaluate whether our $R_0(\boldsymbol{x})$ risk map can predict regions where future cases may occur. Indeed, 2025 WNV cases as of November 2025 were observed in the east, north and southwest of Germany, in agreement with the high risk areas flagged by $R_0(\boldsymbol{x})$. Such a result is important as it helps identify future high-risk regions before seasonal outbreaks occur.
	
	\begin{figure}[H]
		\centering
		\includegraphics[width=1.1\textwidth]{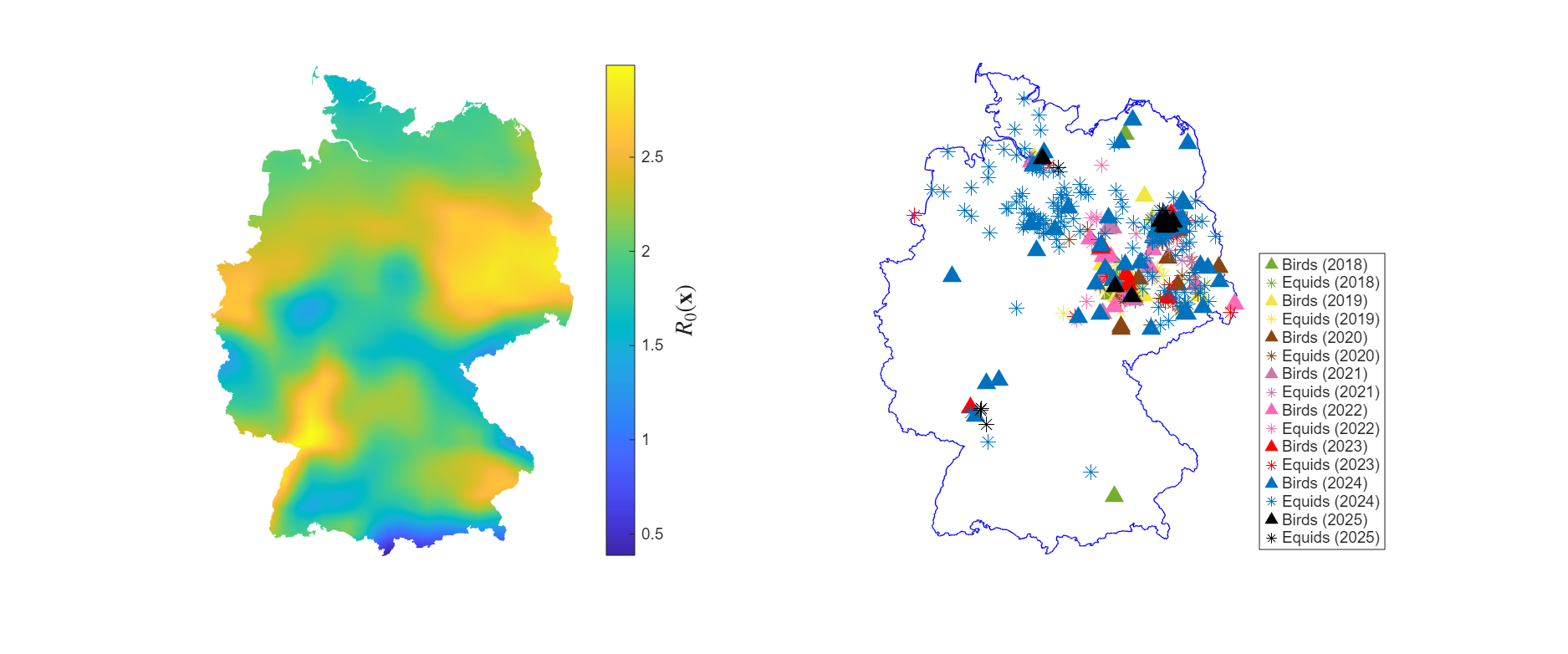}
		\caption{\small The basic reproductive number $R_0(\boldsymbol{x})$ compared to observed data.}
		\label{fb}
	\end{figure}

\section{Simulation framework}\label{rs}
	
Based on available WNV data in Germany, birds of the Accipitridae family were the most affected bird species in 2018 \citep{Ziegler2022}. Therefore, as also done by \cite{Mbaoma2024}, we used parameters associated with northern goshawks and raptors for the resident bird, both of which fall under the Accipitridae family. Initial population densities of infectious birds and mosquitoes are estimated to be high in Halle (Saale), the main WNV hot-spot in 2018, followed by Berlin, and lastly smaller densities in Laage, Poing, and Plessa, as indicated by the 2018 WNV observed cases data \citep{Ziegler2020}. On the other hand, initial densities of susceptible hosts and mosquitoes are assumed to be spatially homogeneously distributed across the entire country. A 2-dimensional Gaussian function of the form 

\begin{equation*}
	f(t,\boldsymbol{x})= e^{\dis -\frac{(x-x_0)^2+(y-y_0)^2}{2\sigma^2}}
\end{equation*}

is applied to define the initial condition peaks. $(x_0,y_0)$ are the central coordinates defining the nodes of the infected regions, and $\sigma$ is the standard deviation. $\sigma=40$ km in Berlin and Halle Saale (main hotspot regions), and 20 km in other infected regions.

Our model is simulated from 1 March 2019 until 30 November 2024 using a daily time step. This is because no WNV activity is observed during December to February, but to account for seasonal migration patterns of migratory birds, we include the spring and autumn months. During the model simulation process, the final solution of each year is saved and imported as initial conditions for the following year. This is done to account for the inflow of new data through ongoing WNV surveillance activities. Moreover, we are able to update initial conditions for the year 2020, where a new case of WNV was confirmed in the north of Germany (Hamburg), and cases in this region may have played a role in the spreading pattern observed in the north of Germany during the subsequent years. In the process, we use inverse modeling to infer the diffusion and advection coefficients that could have contributed to the observed spreading pattern. This is done by running the model with different coefficients while visually comparing the model solution with observed data. The final coefficients are chosen based on their ability to reproduce the observed pattern, while respecting the biology of the magnitudes, i.e., $\dis D_{1}<D_{2}<D_{3}.$ The use of such an approach is mainly because available data is limited and sparse, and thus we are mainly interested in explaining the observed qualitative WNV spreading pattern. 

Temperature strongly influences the activity of \textit{Cx. pipiens} including their competence in transmitting WNV \citep{Ciota2014}. This motivated us to model mosquito parameters using temperature-dependent functions, namely: mosquito biting rate $\beta(T,\boldsymbol{x}),$ latency rate  $\gamma_V(T,\boldsymbol{x}),$ and mortality rate $\mu_V(T,\boldsymbol{x}).$ ERA5 reanalysis data on single levels from 1940 to present is downloaded from the Copernicus Climate Data Store \citep{copcds}. The data is first cleaned, and the years 2019-2024 are subsequently extracted for interpolation onto the mesh and used to compute the temperature-space dependent mosquito functions. During the simulation process, $\beta(T,\boldsymbol{x}), \gamma_V(T,\boldsymbol{x})$  and $\mu(T,\boldsymbol{x})$ are dynamically updated based on the temperatures at different locations in Germany. Advection representing the directed movement of migratory birds is fixed throughout the years. The spring to early summer movement is described by the vector field $[v_x;v_y]=[1;1]$ while the late summer to autumn is described by $[v_x;v_y]=[-1;-1]$ in km per day. To solve our PDE model, we use the Matlab PDEToolbox \citep{MATLAB2023}, which uses a finite element algorithm. The toolbox solves PDEs of the form
	\begin{equation}
		m\frac{\partial^2u}{\partial t^2}+d\frac{\partial u}{\partial t}-\nabla \cdot\left(c\nabla u\right)+au=f,
	\end{equation}
	
and in our case, $m=0, d=1,$ $a=0.$ Matrices $c$ and $f$ capture the diffusion coefficients and reaction terms, respectively, and they are displayed in supplementary (S4), together with the pseudo-code for the simulation procedure. The map of Germany is downloaded as a shape file from \citep{GADM} and simplified using mapshaper \citep{Harrower2006}. The observed data indicates presence data only, i.e., no absence data is available. For this reason, we qualitatively present our results (Low - High) to interpret the model solution. 
	
\section{Model validation}
	
To validate our model predictions, we compared simulated cases with observed data. The validation data comprises bird and equid infection events, as well as mortality, obtained from the \cite{klji}. Due to the lack of a systematic WNV surveillance in Germany, many dead birds go unnoticed, while some are discovered at very advanced stages that affect PCR tests. Domesticated equids are usually easier to diagnose as owners often report their sick animals. Given that birds are the amplifying hosts, it is reasonable to make an assumption that before equid cases are reported, there is an amplifying host that was infected first and might have gone undetected. Moreover, our observed bird data is dominated by dead birds, while the equid data is dominated by infected equids. Because of this, our validation dataset compares merged WNV-related events observed in birds and equids, with the compartment of dead resident birds, given the high WNV-related deaths observed in birds of the Accipitridae family \citep{Mbaoma2024, Ziegler2022}. Firstly, we aggregate both the simulated and observed data within the GADM level-2 administrative units of Germany (N = 356) \citep{GADM}, comprising both rural districts (Landkreise) and urban districts (kreisfreie Städte). The PDE solution $u(x,y,t_f)$ is extracted at the final time step, and the domain $\Omega$ is divided into 365 sub domains $\dis \left\{\Omega_i\right\}_{i=1}^N$ corresponding to the GADM level-2 administrative units, such that $\dis \Omega=\bigcup_{i=1}^{N} \Omega_i.$ For each $\Omega_i,$ the mean simulated number of cases is given by $\dis \frac{1}{|\Omega_i|}\int_{\Omega_i} u(x,y,t_f)\; dA,$ where $|\Omega_i|$ is the area of the region $\Omega_i.$

Observed cases are aggregated using the formula $\dis \sum_{j=1}^{N_\text{obs}} \mathbf{1}_{\Omega_i}(x_j, y_j,t_f),$ where $\mathbf{1}_{\Omega_i}(x_j, y_j,t_f)$ is the indicator function equal to 1 if the observed point lies within $\Omega_i$, and 0 otherwise, and $N_\text{obs}$ is the number of observed points per region.
	
Boundaries are shrunk inwards using the \texttt{polybuffer} function \citep{MATLAB2023} from the toolbox to prevent misclassifications near borders. The \texttt{isinterior} operator \citep{MATLAB2023} is then applied to prevent double-counting elements located on regional borders. Thus, only points lying strictly within the interior of a polygon are counted as belonging to that region.
	
Lastly, the Spearman's rank correlation coefficient $(\rho)$ \citep{Spearman1904} is applied to quantify the degree of spatial correspondence between simulated and observed patterns in each district. Spearman's $(\rho)$ is calculated using the formula:

		\begin{equation}
			\rho = 1 - \frac{6 \sum_{i=1}^{N} d_i^2}{N(N^2 - 1)},
	\end{equation}
	
where $d_i$ is the difference between the rank of the simulated and observed points, and $N$ is the number of administrative units. A positive $\rho$ indicates that locations with a high density of simulated cases correspond to locations with a high density of observed points. Conversely, a negative $\rho$ implies that high values in one pattern correspond to low values in the other, indicating an opposite relationship. A $\rho$ close to zero suggests the absence of a monotonic spatial relationship \citep{Spearman1904}. P-values are computed to assess whether the observed correlation occurred by chance, with smaller p-values (<0.05) indicating that the observed correlation is unlikely to have occurred by chance.

\section{Results}\label{fd}
	
The model initialized with 2018 WNV cases data shows strong visual agreement with observed patterns in Germany throughout the 2019-2024 simulation period (Figure \ref{f2}). Its notable strength is the ability to identify key WNV circulation areas which are initially concentrated in the Eastern regions. The RMSD estimates indicate that WNV-mosquito infections reached an average annual dispersal range of 33.1 km in the years 2019-2022. A higher dispersal range of 46.8 km was observed in years 2023-2024 when the virus range expanded to almost the entire country. In migratory birds, an annual RMSD of 104.7 km was observed for all the years except the year 2021 where an RMSD of 93.6 km was observed. On the other hand, a constant annual dispersal range of 74.0 km was observed in resident bird infections for all the years (Table \ref{t03}). The RMSD of mosquitoes remained lower than that of resident and migratory birds, while that of resident birds was lower than that of migratory birds, reflecting the greater mobility of avian hosts compared to mosquito species. 

Spearman's correlation coefficients $(\rho)$ were all statistically significant with the years 2020-2021 showing values greater than 0.5. The other years had values slightly below 0.5, a result partly attributed to the sparsity of the observed data. The corresponding p-values were all below 0.05 (Table \ref{t03}, Figure \ref{g2}). In 2023, the model accurately predicted an isolated case in the Northwestern region (next to the Dutch border) and another isolated case in the Southwest region between the states of Baden-Württemberg, Hessen and Rheinland-Pfalz (Figure \ref{fl01}). This hints at the model's ability to anticipate under-reported spread patterns mainly driven by temperature, diffusion and advection processes (Figure \ref{f2}). During the same year, the model simulated without migratory birds failed to account for these isolated cases (Figure \ref{fl03}). This shows the significance of incorporating migratory bird species and their seasonal movements in our model as this significantly improved the model fit (Figure \ref{fl01}). In general, the spread pattern shows a progressive, wave-like expansion mainly dominated in the eastern regions initially. The wave then travels in an anti-clockwise direction to other parts of Germany. The model's ability to explain the observed pattern suggests that incorporating temperature, diffusion and migratory processes in PDE-based WNV models significantly enhances their predictive capacity.

	\begin{figure}[H]
		\centering
		\begin{subfigure}{0.35\textwidth}
			\centering
			\includegraphics[width=\textwidth]{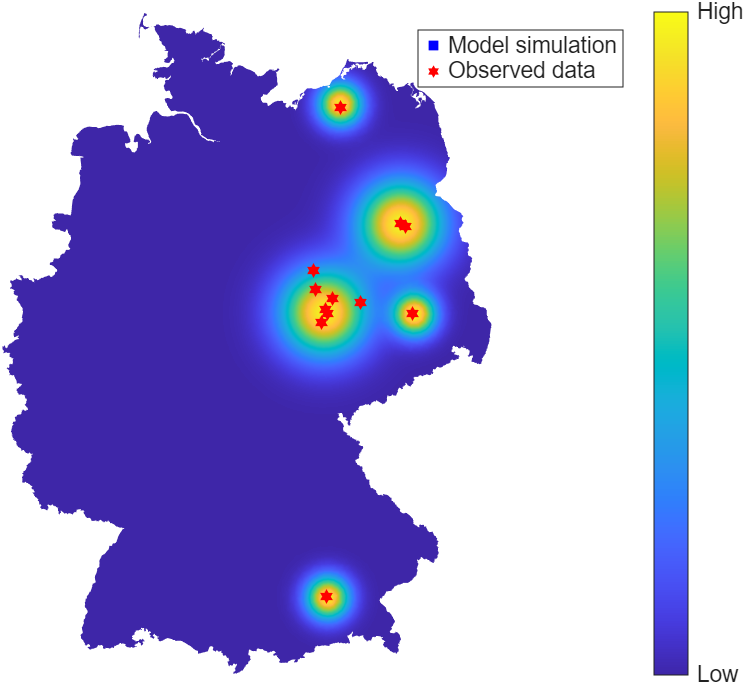}
			\caption{\small 2018 cases.}
			\label{f00}
		\end{subfigure}
		\hfill
		\begin{subfigure}{0.35\textwidth}
			\centering
			\includegraphics[width=\textwidth]{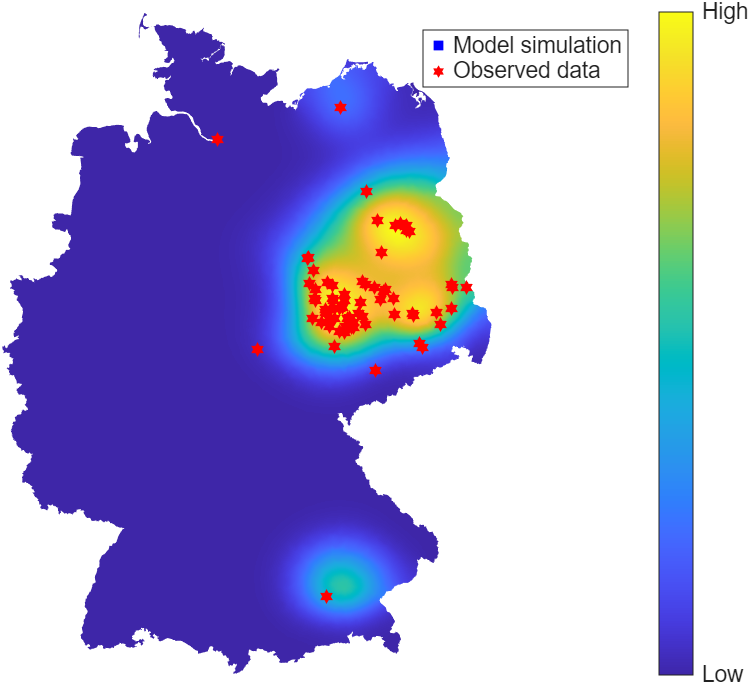}
			\caption{\small 2018-2019 cases.}
			\label{f01}
		\end{subfigure}
		\hfill
		\begin{subfigure}{0.35\textwidth}
			\centering
			\includegraphics[width=\textwidth]{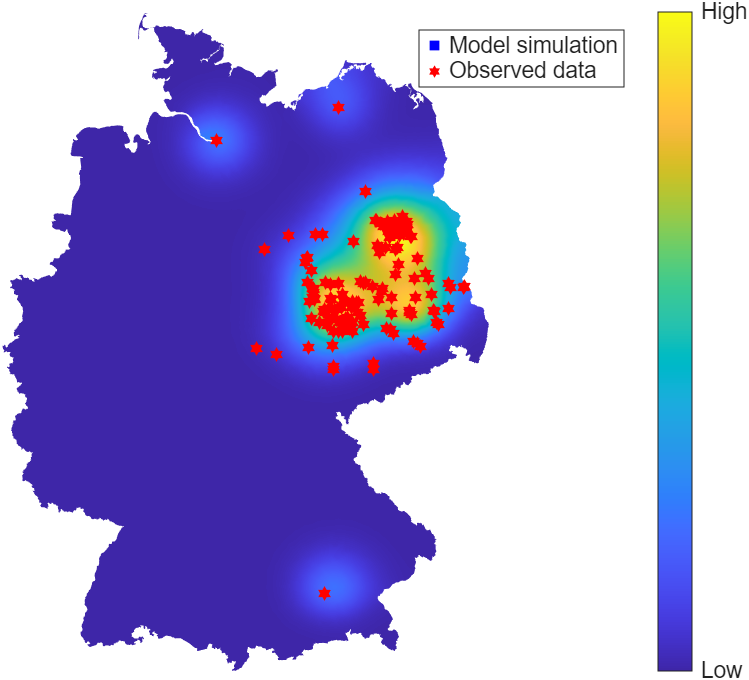}
			\caption{\small 2018-2020 cases.}
			\label{f03}
		\end{subfigure}
		\hfill
		\begin{subfigure}{0.35\textwidth}
			\centering
			\includegraphics[width=\textwidth]{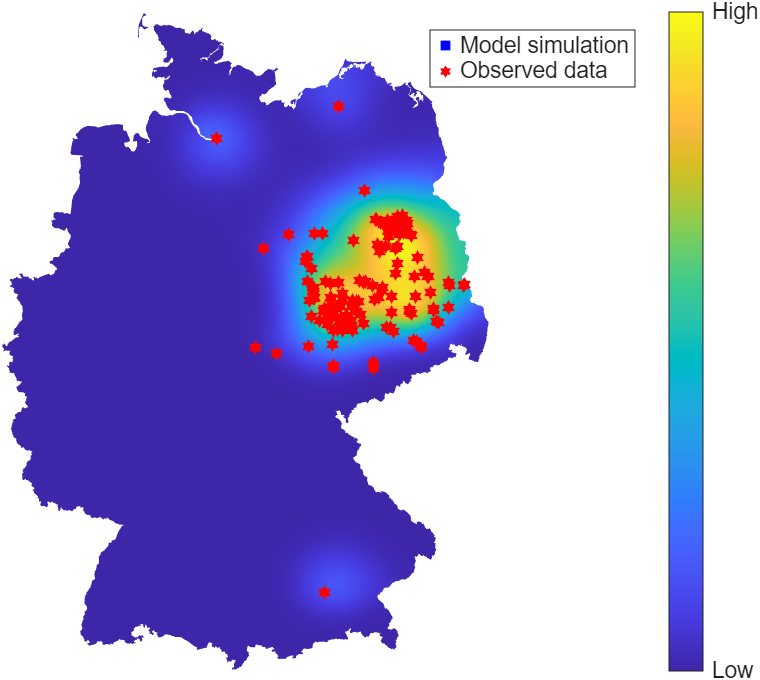}
			\caption{\small 2018-2021 cases.}
			\label{f05}
		\end{subfigure}
		\hfill
		\begin{subfigure}{0.35\textwidth}
			\centering
			\includegraphics[width=\textwidth]{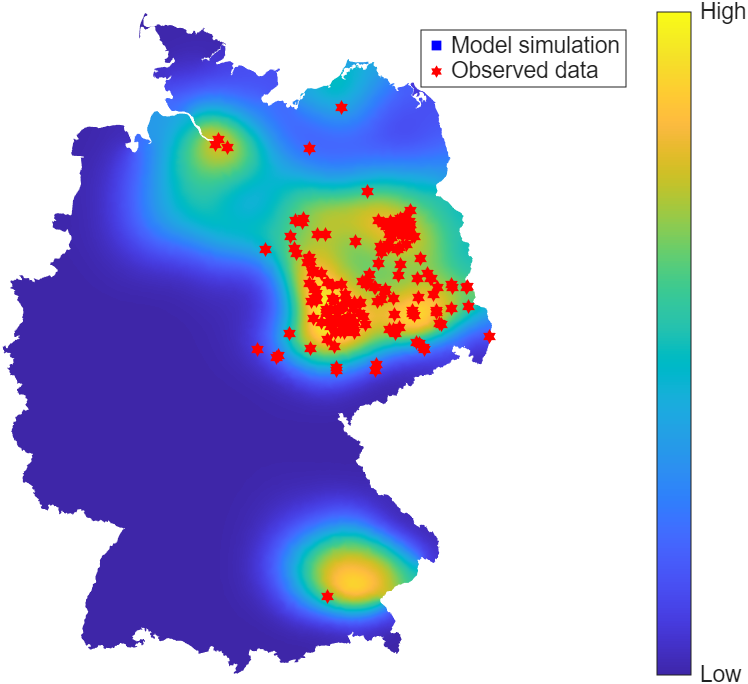}
			\caption{\small 2018-2022 cases.}
			\label{f06}
		\end{subfigure}
		\hfill
		\begin{subfigure}{0.35\textwidth}
			\centering
			\includegraphics[width=\textwidth]{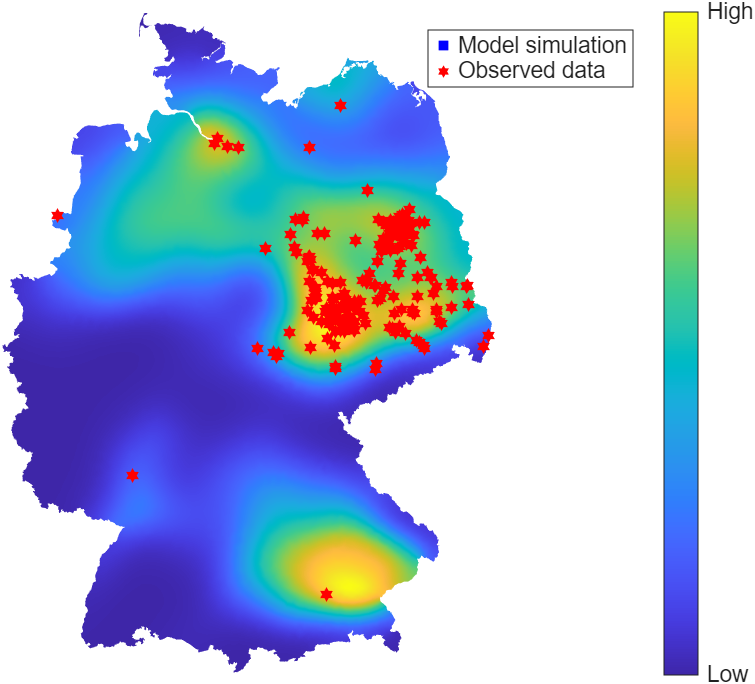}
			\caption{\small 2018-2023 cases.}
			\label{f07}
		\end{subfigure}
		\hfill
		\begin{subfigure}{0.35\textwidth}
			\centering
			\includegraphics[width=\textwidth]{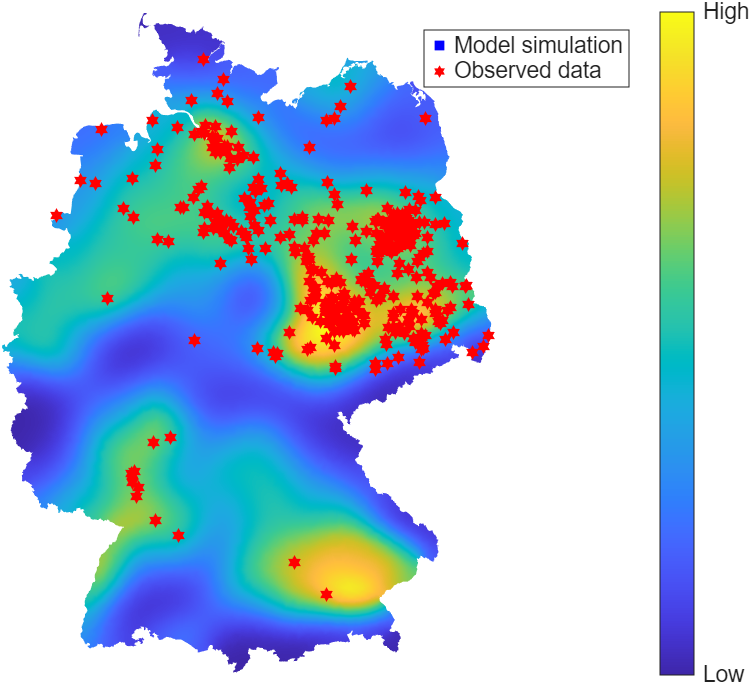}
			\caption{\small 2018-2024 cases.}
			\label{f08}
		\end{subfigure}
		\caption{\small Model simulation compared to observed data of WNV in Germany between the years 2018-2024.}
		\label{f2}
	\end{figure}

	\begin{table}[H]
		\caption{\small Summary values for the diffusion coefficients $(D_1,D_2,D_3),$ in km$^2$ per day; RMSD for each species in km; Spearman's correlation coefficient and P-values.}
		\renewcommand{\arraystretch}{1.1}
		\small
		\centering
		\begin{adjustbox}{width=1.0\textwidth}
			\begin{tabular}{lllllllll}
				\toprule 	
				\textbf{Year}&D$_1$&D$_2$ &D$_3$&RMSD$_1$&RMSD$_2$&RMSD$_3$&$\rho$& P-value\\
				\midrule 
				2019&1.0&5.0&10.0&33.1&74.0&104.7&0.41&$1.6\times 10^{-15}$\\
				2020&1.0&5.0&10.0&33.1&74.0&104.7&0.53&$6.4\times 10^{-26}$\\
				2021&1.0&5.0&8. 0&33.1&74.0&93.6 &0.51&$1.1\times 10^{-23}$\\
				2022&1.0&5.0&10.0&33.1&74.0&104.7&0.47&$3.6\times 10^{-20}$\\
				2023&2.0&5.0&10.0&46.8&74.0&104.7&0.43&$1.5\times 10^{-16}$\\
				2024&2.0&5.0&10.0&46.8&74.0&104.7&0.35&$1.2\times 10^{-11}$\\
				\bottomrule
			\end{tabular}
		\end{adjustbox}
		\label{t03}
	\end{table}

	\begin{figure}[H]
		\begin{subfigure}{0.4\textwidth}
			\centering
			\includegraphics[width=\textwidth]{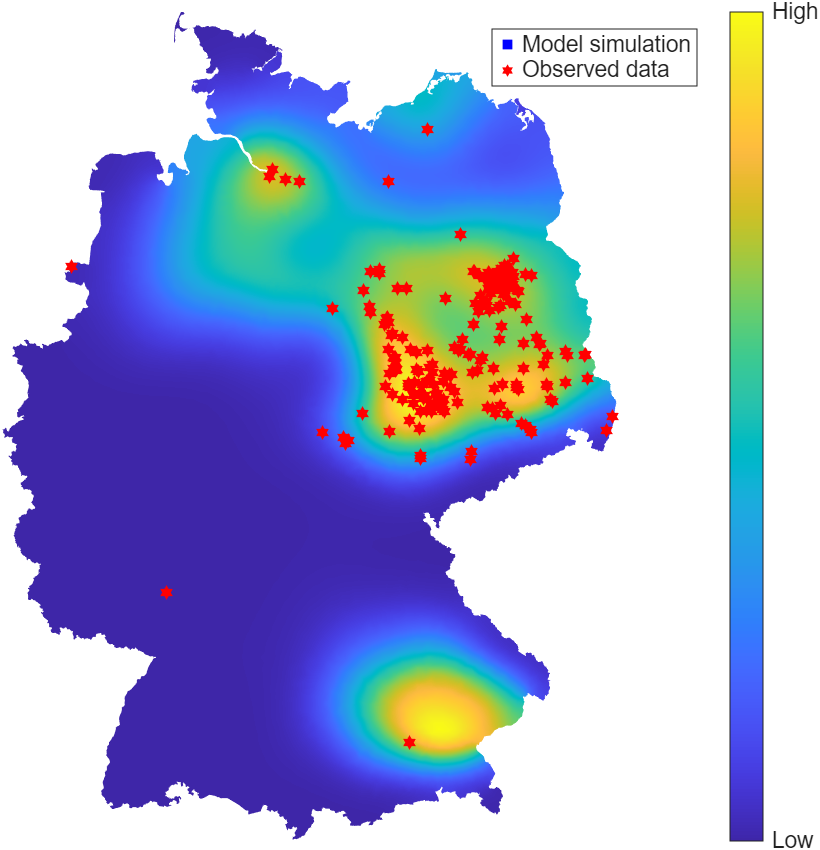}
			\caption{\small Model simulation without migratory birds in 2023.}
			\label{fl03}
		\end{subfigure}
		\hfill
		\begin{subfigure}{0.45\textwidth}
			\centering
			\includegraphics[width=\textwidth]{uv6.png}
			\caption{\small Model simulation with migratory birds in 2023.}
			\label{fl01}
		\end{subfigure}
		\caption{\small Contribution of migratory birds to the WNV spread pattern during the year 2023.}
		\label{fl2}
	\end{figure}

	\begin{figure}[H]
		\centering
		\begin{subfigure}{0.45\textwidth}
			\centering
			\includegraphics[width=\textwidth]{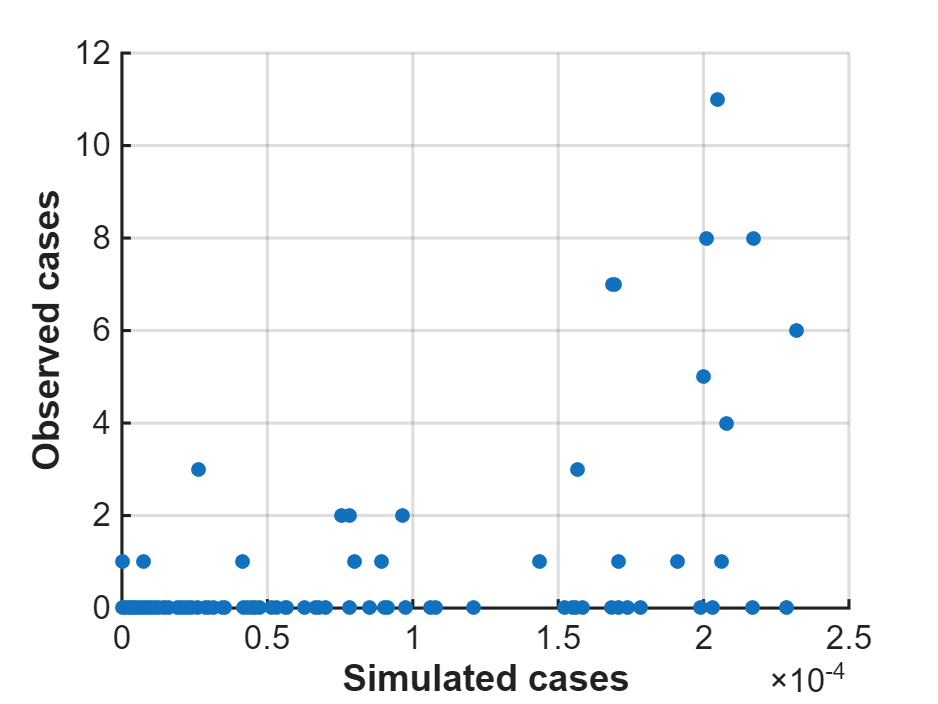}
			\caption{\small 2018-2019 cases.}
			\label{g01}
		\end{subfigure}
		\hfill
		\begin{subfigure}{0.45\textwidth}
			\centering
			\includegraphics[width=\textwidth]{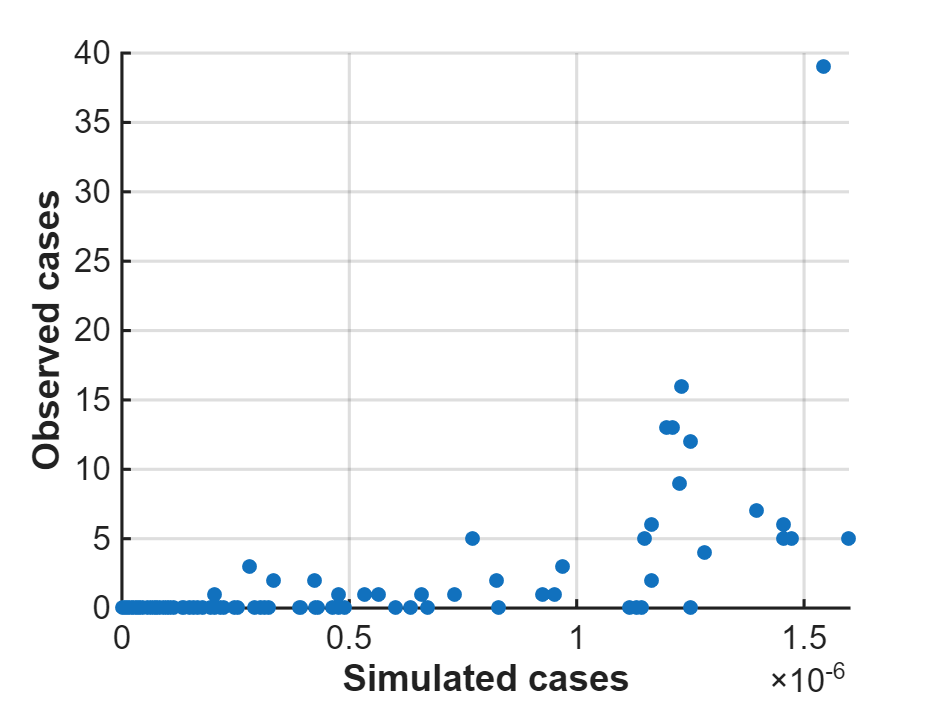}
			\caption{\small 2018-2020 cases.}
			\label{g03}
		\end{subfigure}
		\hfill
		\begin{subfigure}{0.45\textwidth}
			\centering
			\includegraphics[width=\textwidth]{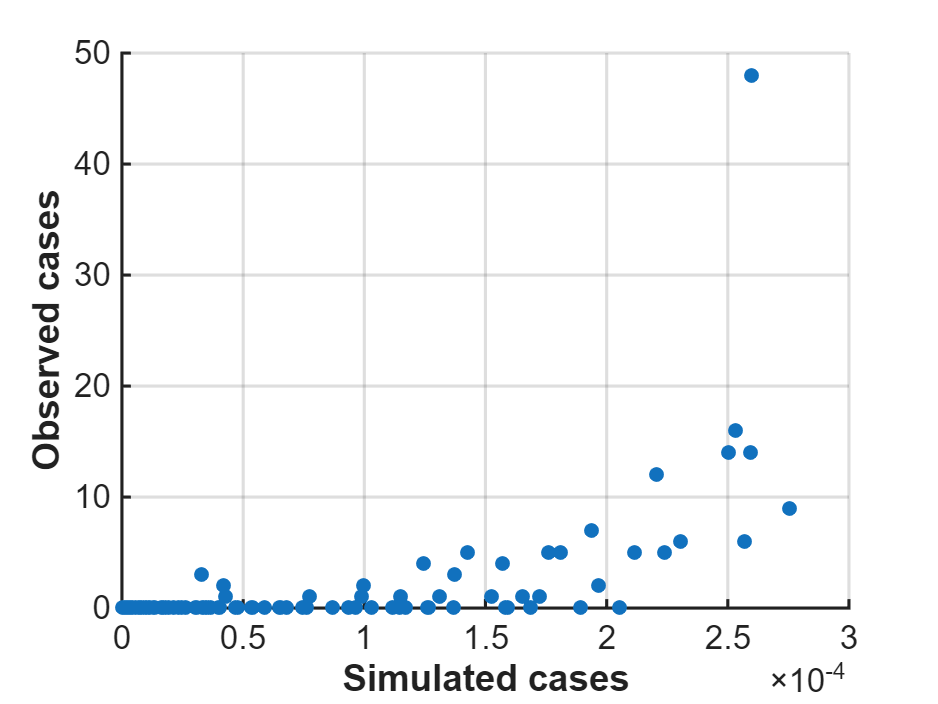}
			\caption{\small 2018-2021 cases.}
			\label{g05}
		\end{subfigure}
		\hfill
		\begin{subfigure}{0.45\textwidth}
			\centering
			\includegraphics[width=\textwidth]{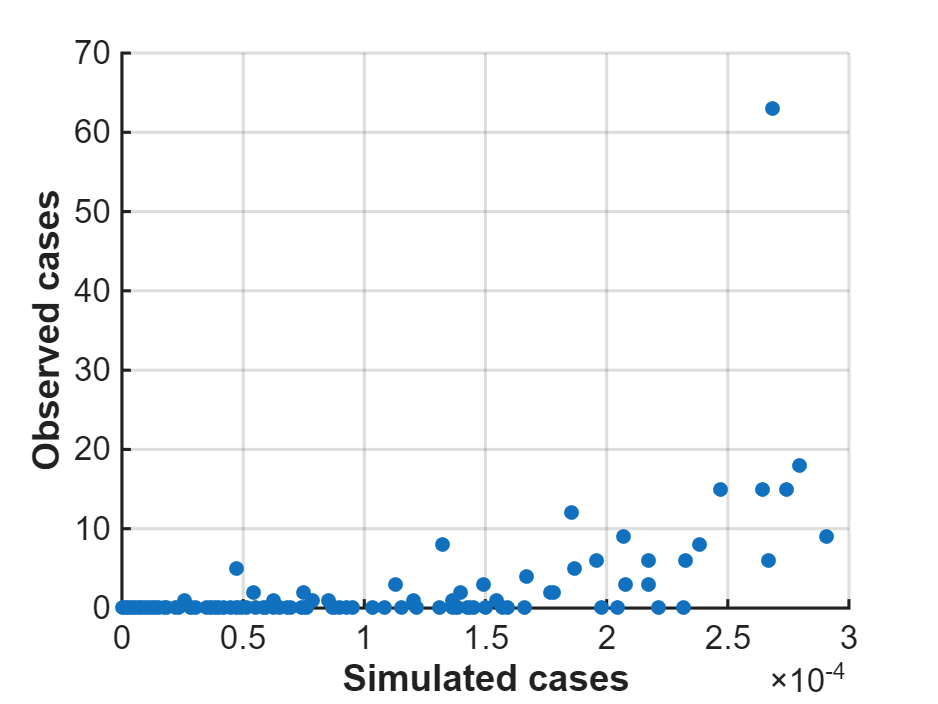}
			\caption{\small 2018-2022 cases.}
			\label{g06}
		\end{subfigure}
		\hfill
		\begin{subfigure}{0.45\textwidth}
			\centering
			\includegraphics[width=\textwidth]{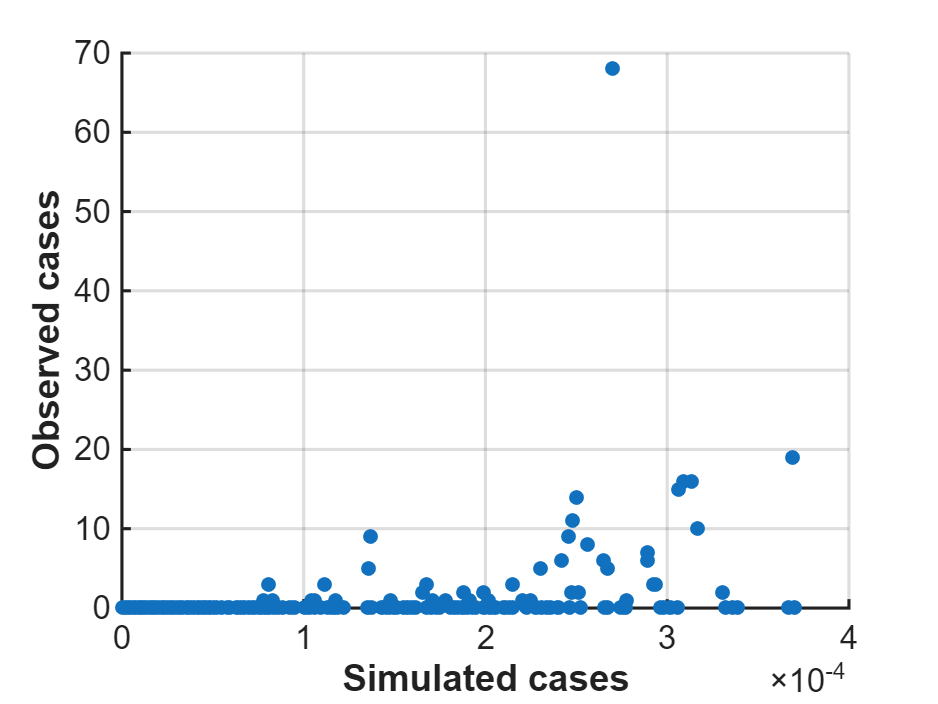}
			\caption{\small 2018-2023 cases.}
			\label{g07}
		\end{subfigure}
		\hfill
		\begin{subfigure}{0.45\textwidth}
			\centering
			\includegraphics[width=\textwidth]{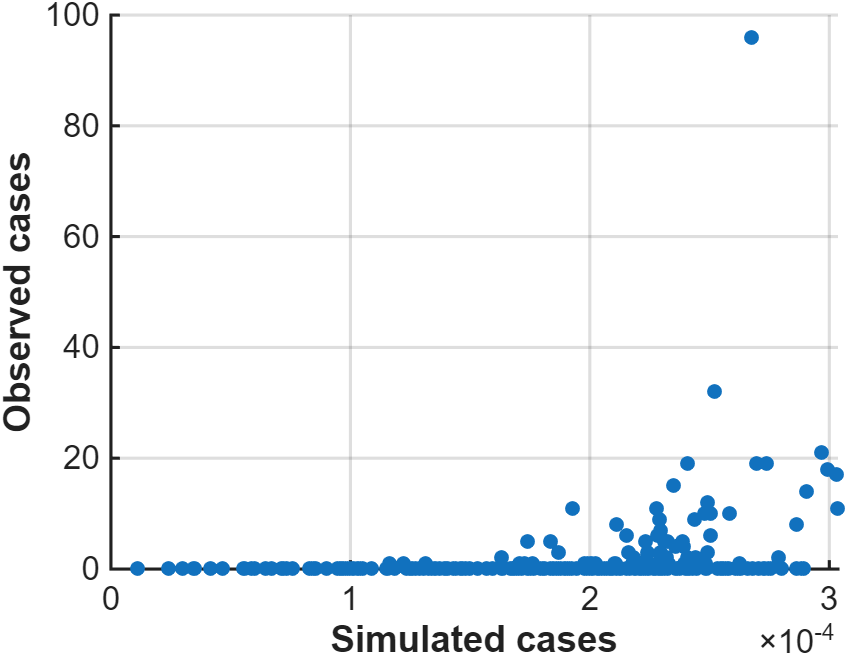}
			\caption{\small 2018-2024 cases.}
			\label{g08}
		\end{subfigure}
		\caption{\small Scatter plots for Spearman's rank correlation showing the monotonic relationship between the model prediction and observed data.}
		\label{g2}
	\end{figure}

\section{Discussions and conclusions}
	
Understanding the WNV spreading pattern allows early detection of WNV outbreaks. The identification of circulation areas also helps early planning and resource allocation for preventive measures. At the same time, it is costly to carry out systematic surveillance programs aiming to understand the spread of WNV. Therefore, in this study, we modeled the spread of WNV in Germany using a temperature-driven PDE model with advection and diffusion terms. We proved that our model has non-negative and unique solutions that are bounded, and that the system (\ref{sys2a}) is well-posed.
	
In agreement with previous studies \citep{Bhowmick2020, DBergsman2016, Mbaoma2024}, at the temporal level, the basic reproductive number $(R_0(\boldsymbol{x}))$ indicates that the presence of migratory birds increases the number of amplifying local hosts. At the spatial scale, $R_0(\boldsymbol{x})$ identified low-high risk regions of public health interest, which agreed with the observed WNV cases in birds and equids, and other temperature-driven WNV models, e.g., by \cite{Mbaoma2024}. The reproductive number classified regions in the eastern parts of Germany as the main hotspots, a result similar to that obtained by \cite{Bhowmick2023} and \cite{Mbaoma2024}. In addition, $R_0(\boldsymbol{x})$ flagged regions in the north, south and southwest of Germany as new circulation areas. Our approach explains the drivers of the WNV spreading pattern in Germany, while inferring estimates for spreading speeds. In particular, from the simulation of our PDE model, we computed the annual RMSD, which estimates the average displacement of the infection after a certain period of time. Due to a lack of real movement data for hosts, we chose diffusion coefficients that made the model prediction match the observed data.
	
During the year 2023, the model prediction showed a poor fit to observed data when only resident birds were considered. However, incorporating migratory bird compartments and capturing their directional seasonal flyways led to an improved fit for the same year. Observed WNV spreading pattern aligns with known migratory bird flyways, suggesting that directed long-range bird migration could play a significant role in disseminating the WNV. This observation is also supported by a recent study by \cite{Tth2025}. Our results suggest that migratory birds may still have played an indirect role in shaping the WNV spreading pattern in Germany, even if they may not be the primary source of infection. Thus, effective surveillance and prevention efforts of WNV should also be implemented at key stopover sites, where migratory birds are likely to rest and interact with local mosquitoes.

 The observed data for 2023 did not show any WNV cases between the northern region surrounding Hamburg and the isolated case near the Dutch border. However, our model flagged this region as a potential circulation region, a claim that was validated by the WNV data for 2024 (Figure \ref{fb} and Figure \ref{f07}). This highlights the predictive capacity of our model, mainly driven by the diffusion, advection and temperature. Furthermore, our model uncovers a transmission corridor for WNV that spreads from the east in an anti-clockwise direction, likely reflecting the temperature gradient that supports the long-term establishment of highly competent \textit{Cx.} mosquitoes. Interestingly, the circulation areas identified in our model agree with those in the study by \cite{DiPol2022}, which aimed at modeling the temperature suitability for the risk of WNV establishment in European \textit{Cx. pipiens} populations. As of November 2025, WNV bird and equid cases for the year 2025 had been reported in the the Eastern, Northern and Southwestern regions, regions that were previously flagged as circulation areas in previous years (Figure \ref{fb}, right).
	
Results from our spatially explicit PDE model demonstrate that temperature, diffusion and advection, together with resident and migratory birds have a significant influence on the spreading pattern of WNV in Germany. However, despite our model's ability to simulate the spread of WNV in Germany, some additional factors can be considered to improve such PDE models in the future. For example, land use is known to influence the general spatio-temporal distribution of mosquito populations \citep{Watts2021}, which plays a huge role in the mosquito abundance. Furthermore, we used temperature as the only weather variable to explain the heterogeneity of our mosquito parameters. Whereas, precipitation, wind speed and direction also influence the general distribution and activity of mosquitoes \citep{Ducheyne2007}. The estimation of diffusion and advection coefficients by inverse modeling causes uncertainties in the results of the model, and thus, this work can be improved by considering more systematic methods for inferring diffusion and advection coefficients, such as Bayesian methods \citep{BarajasSolano2019}. To further capture the spatial heterogeneity in the model, another improvement would be the use of spatially varying diffusion coefficients, given that in reality, movement rates are not constant. Our PDE model can be extended to a stochastic PDE system to study more random effects that cannot be modeled by deterministic PDEs. The current WNV surveillance system remains patchy, and unable to detect the virus early before a spillover to humans occurs. This is evidenced by that most data arises from incidental or passive detection rather than systematic monitoring \citep{Constant2022}. Such a system poses a significant risk of a reporting bias in the observed data, with a potential of under-reporting of cases, prompting our results to be considered as minimum estimates, rather than the exact WNV situation.

\section*{CRediT authorship contribution statement}
	
\textbf{Pride Duve:} Designed the study, developed the model, proved the mathematical properties, computed numerical simulations, and wrote the original draft. \textbf{Felix Sauer:} Supervised the study, designed the model, and reviewed the final draft. \textbf{Renke Lühken:} Supervised the study, designed the model, and reviewed the final draft.

\section*{Declaration of competing interests}
	
The authors declare that they have no known competing financial interests or personal relationships that could have appeared to influence the work reported in this paper.
	
\section*{Data availability}
	
The cleaned WNV data for birds and equids used in this manuscript is freely available  \href{https://github.com/Pride1234/West-Nile-Virus-Files.git}{\colorbox{brown!10}{\textcolor{red}{\textbf{Here}}}}. Original datasets can be obtained from the \cite{klji} website.

\section*{Acknowledgements}
	
	The authors would like to thank the Federal Ministry of Education and Research of Germany (BMBF) under the project NEED (Grant Number 01Kl2022), the Federal Ministry for the Environment, Nature Conservation, Nuclear Safety and Consumer Protection (Grant Number 3721484020), and the German Research Foundation (JO 1276/51) for funding this project. The authors are also grateful to the National Research Platform for Zoonosen for funding the research visit of Pride Duve to Prof. Dr. Reinhard Racke's lab, at the University of Konstanz. Furthermore, we would like to acknowledge and thank Reinhard Racke, Nicolas Schlosser, and Stefanie Frei for the discussions, consultations, and support.

	\bibliographystyle{plainnat}  
	\bibliography{References}      
	
\end{document}